\documentclass[12pt]{article}

\usepackage{amssymb, amsmath,amsfonts, bm, eurosym,geometry,graphicx,caption,mathtools,natbib, xcolor,setspace,comment,pdflscape,array, parskip, bm, csquotes, multirow, booktabs, bbm, nicefrac, authblk}

\usepackage[normalem]{ulem}

\usepackage[]{footmisc}
\usepackage[colorlinks=true]{hyperref}
\usepackage{footnotebackref}
\usepackage{dirtytalk}
\usepackage[]{threeparttable}

\usepackage[capposition=bottom]{floatrow}

\usepackage[toc,page]{appendix}

\usepackage{color, amsthm}
\usepackage{tikz, subfig, cleveref, istgame, xparse, expl3}
\usetikzlibrary{positioning, patterns, decorations.pathreplacing, calligraphy, arrows.meta, bending, decorations.markings, external}
\tikzset{
    arc arrow/.style args={%
    to pos #1 with length #2 and options #3}{
    decoration={
        markings,
         mark=at position 0 with {\pgfextra{%
         \pgfmathsetmacro{\tmpArrowTime}{#2/(\pgfdecoratedpathlength)}
         \xdef\tmpArrowTime{\tmpArrowTime}}},
        mark=at position {#1-\tmpArrowTime} with {\coordinate(@1);},
        mark=at position {#1-2*\tmpArrowTime/3} with {\coordinate(@2);},
        mark=at position {#1-\tmpArrowTime/3} with {\coordinate(@3);},
        mark=at position {#1} with {\coordinate(@4);
        \draw[-{Stealth[length=#2,bend,#3]}]
        (@1) .. controls (@2) and (@3) .. (@4);},
        },
     postaction=decorate,
     }
}
\tikzstyle{line}=[draw]
\usepackage{pgfplots, tablefootnote}
\usepgfplotslibrary{fillbetween}
\pgfplotsset{compat=1.3}

\graphicspath{ {./images/} }

\newtheorem{assumption}{Assumption}
\newtheorem{proposition}{Proposition}
\newtheorem{lemma}[proposition]{Lemma}
\newtheorem{corollary}[proposition]{Corollary}

\newtheorem{definition}{Definition}

\newcommand{\hatalp}{\hat \alpha}

\newcolumntype{L}[1]{>{\raggedright\let\newline\\arraybackslash\hspace{0pt}}m{#1}}
\newcolumntype{C}[1]{>{\centering\let\newline\\arraybackslash\hspace{0pt}}m{#1}}
\newcolumntype{R}[1]{>{\raggedleft\let\newline\\arraybackslash\hspace{0pt}}m{#1}}

  {\list{}{\leftmargin=0.3in\rightmargin=0in}\item[]}%
  {\endlist}
  
\hypersetup{
    pdftoolbar=true,        
    pdfmenubar=true,        
    pdffitwindow=false,     
    pdfstartview={FitH},    
    pdftitle={title},    
    pdfauthor={Salvatore Mazzarino},     
    pdfsubject={Subject},   
    pdfcreator={Salvatore Mazzarino},   
    pdfproducer={Salvatore Mazzarino}, 
    pdfkeywords={Green Networking} {Mobile Cloud} {Network Coding} {Energy}, 
    pdfnewwindow=true,      
    colorlinks=true,       
    linkcolor=blue,          
    citecolor=blue,        
    filecolor=blue,      
    urlcolor=blue           
}

\makeatletter
\renewcommand\hyper@natlinkbreak[2]{#1}
\makeatother

\begin{document}

\title{\vspace{-0.8cm} 
Misspecified beliefs and the evolution of peer pressure\thanks{We are grateful to Pierpaolo Battigalli, Alberto Dalmazzo, Sebastiano Della Lena, and Philip Neary for helpful discussions and suggestions. We also thank seminar participants at the Psychology Seminar Series at Royal Holloway, University of London, for useful comments. Financial support from Leverhulme International Professorship (grant number LIP-2022-001) is gratefully acknowledged.}}

\author[s,b]{Paolo Pin}
\author[r]{Roberto Rozzi}
\affil[s]{Dipartimento di Economia Politica e Statistica, Universit\`a di Siena, Italy}
\affil[b]{BIDSA,  Universit\`a  Bocconi, Milan, Italy}
\affil[r]{Royal Holloway, University of London}

\date{May 2026}

\maketitle

\begin{abstract}
\noindent We study the emergence of conformity preferences in an environment in which agents choose effort under heterogeneous, possibly misspecified returns, and social interactions do not directly affect material payoffs. Some agents choose effort by trading off performance and conformity to expected peer behavior. We characterize subjective best responses. For any given beliefs, an optimal and unique level of peer pressure exists and is evolutionarily stable within groups of agents sharing the same misspecification. Such a level is zero for correctly specified agents and may be positive for misspecified ones. When the efficient level of peer pressure is interior, misspecified agents choose effort equal to their true return, resulting in an equilibrium behavior that is both self-confirming and Nash, allowing the persistence of misspecifications. Peer pressure need not generate long-run allocative distortions, but it creates a perceived value of social information. In equilibrium, this value depends only on misspecification, generating scope for informational rents. 


\noindent \\
\vspace{-0.75cm} 
\vspace{0in}

\end{abstract}
\setcounter{page}{0}
\thispagestyle{empty}

\pagebreak 

\section{Introduction}\label{sec:intro}

Humans often feel the urge to conform with their peers in a wide range of contexts \citep{mascagni2018taxreview,farrow2017social}. Susceptibility to peer influence appears to be particularly pronounced during adolescence \citep{steinberg2007age,laursen2022does}. Despite the large empirical literature documenting conformity in behavior, much less is understood about the mechanisms through which conformity motives arise in the first place. We provide a simple explanation for this pattern by studying a model in which agents' material utility does not directly depend on the actions of others, yet it may still be optimal for them to behave as if social pressure were payoff--relevant.

We consider a population of agents partitioned into groups that differ in their returns to effort. Within each group, some agents know their own return to effort, while others do not. Agents' material utility does not depend on what their peers do; however, some agents believe that it does. As a result, they choose their effort based on a perceived utility function that incorporates the expected effort of their peers. Specifically, for these agents, perceived utility is a weighted average of true material utility and a conformity component capturing the distance between their own effort and the expected effort of their peers.
In this respect, misspecified agents in our model resemble users of a false but predictively successful theory: they may behave optimally even though their internal representation of the environment is incorrect.\footnote{A useful historical analogy is Ptolemaic astronomy. Although geocentric and ultimately false, Ptolemy's system achieved substantial predictive accuracy in practice. At the same time, it is often contrasted with later frameworks as relying on a more elaborate corrective structure rather than on a correct account of the underlying system; see, e.g., \citet{wray2018false}.}

We first characterize agents' subjective best response under this perceived utility. We then derive the efficient level of social pressure, defined as the level that maximizes material utility given the subjective best response, and assess the robustness of such a level to invasions by agents who share the same misspecification but differ in the intensity of social pressure they attach to peer behavior.

We find that the subjective best response always exists and is unique, and that the efficient level of peer pressure is also uniquely determined. In particular, for correctly specified agents the efficient level is always zero, whereas for misspecified agents it may be interior, that is, strictly between $0$ and $1$. We further show that, whenever the efficient level of social pressure is interior, the corresponding subjective best response is both a self-confirming equilibrium and a Nash equilibrium. It is self-confirming because, along the induced behavior, agents do not receive feedback that contradicts their beliefs; it is a Nash equilibrium because, at the efficient level of social pressure, the subjective best response also maximizes material utility. The implication is that agents may behave optimally without learning their true returns: misspecified beliefs can persist even when behavior is fully efficient. This also suggests a broader interpretation of the model: conformity concerns can survive not because they are directly payoff-relevant, but because they help agents act efficiently in environments with imperfect self-knowledge.


Furthermore, the efficient level of peer pressure depends on three key features: the structure of peer composition, the fraction of misspecified agents, and the degree of misspecification within each group. In particular, greater homophily in the network---that is, a higher weight placed on peers with similar returns to effort---increases the efficient level of peer pressure when peers behave optimally. Intuitively, when peers are more informative about an agent's own environment, for agents with incorrect beliefs, conforming to their behavior becomes a more effective way to compensate for misspecification. This effect may be attenuated or even reversed when peers themselves behave suboptimally, due to equilibrium feedback effects.

The impact of the fraction of misspecified agents on peer pressure operates through the informativeness of peer behavior. When peers behave optimally, this fraction does not matter, since all peers exert the optimal level of effort. When they do not, its impact is conditional on whether changes in composition make peers more or less informative about the group’s true return to effort. The relationship between peer pressure and misspecification is instead conditional on the group's relative position. The efficient level of peer pressure increases with misspecification if and only if the group's return to effort exceeds the average effort in its peer network. In this case, placing more weight on peers helps correct individual misperceptions; otherwise, stronger misspecification reduces the value of conformity.

Given the central role of peers' behavior in agents' decisions, we study the value that agents attribute to observing such behavior. Conditioning behavior on peer outcomes requires attention and computation, and may generate adjustment costs when beliefs are misspecified, even if behavior ultimately coincides with the efficient benchmark. More importantly, within the model agents perceive a direct benefit from observing the social statistic summarizing peers' actions, because this allows them to condition their behavior on the prevailing environment. In equilibrium, however, the perceived value of monitoring depends only on the gap between agents' misspecified beliefs about fundamentals and the true underlying returns. This implies that environments with larger misperceptions generate a higher perceived value of social information, creating scope for informational intermediaries---such as platforms providing access to social signals---to extract informational rents.

The remainder of the paper is organized as follows. After a brief review of the related literature, Section~\ref{sec:mod} introduces the model. Section~\ref{sec:res} presents the main results, while Section~\ref{sec:policy} studies the informational value and costs of observing peers' behavior. Section~\ref{sec:disc} discusses the broader relevance of the mechanism and its connection to two strands of the literature. Section~\ref{sec:concl} concludes. Proofs are collected in Appendix~\ref{secA_proofs}, Appendix~\ref{secB:res_2G} provides additional explicit results for the two-group case, and Appendix~\ref{secC:general_payoffs} extends the analysis to general strictly concave payoffs.

\paragraph{Literature review.}

Our study can be viewed as an indirect evolutionary account of peer pressure. More specifically, we use indirect evolutionary game theory to study the emergence of conformity preferences. Among the contributions analyzing departures from purely material preferences, \citet{bernheim1994theory} is a seminal reference emphasizing the role of conformity. In Bernheim's model, conformity preferences are taken as primitive, but are motivated by an evolutionary argument: individuals who are more highly regarded may enjoy greater reproductive opportunities.

While the mechanism proposed by \citet{bernheim1994theory} is plausible, it is also consistent with agents correctly anticipating that social approval may affect material payoffs, rather than with preferences that depart from material incentives. Our framework differs from \citet{bernheim1994theory} since we explicitly introduce misspecified beliefs. Within this environment, we show that the evolutionarily efficient intensity of social pressure depends on both the degree of misspecification and the structure of the peer network, in particular the extent of assortative interaction across types.

More generally, the literature on the evolution of preferences has its roots in earlier contributions such as \citet{becker1976altruism}, \citet{hirshleifer1977economics}, \citet{rubin1979evolutionary}, and \citet{frank1987if}.\footnote{For broad reviews of this literature, see \citet{alger2019evolutionary} and \citet{alger2023evolutionarily}.} The formalization of the indirect evolutionary approach is usually traced back to \citet{guth1992explaining} and \citet{guth1995evolutionary}. Early work in this literature focused on how the observability of preferences shapes evolutionary stability, with important contributions by \citet{ely2001nash}, \citet{ok2001evolution}, \citet{sethi2001preference}, \citet{heifetz2007maximize}, and \citet{dekel2007evolution}. More recent work has instead emphasized the role of the matching process in shaping evolutionary outcomes \citep[e.g.,][]{alger2013homo,alger2016evolution,alger2020evolution,wu2020labelling,wu2021preference}.

Our model contributes to this literature by introducing an additional dimension to the interaction structure: interaction is assortative in fundamentals (i.e., return to effort) rather than in preferences. In this setting, we show that the evolutionarily efficient intensity of conformity increases with the degree of assortative interaction across types. More broadly, our results clarify the environments in which conformity concerns are more likely to emerge and the individuals for whom such concerns are likely to be stronger. While our analysis is conducted in a specific applied setting, it highlights a novel channel through which assortativity in fundamentals—distinct from assortativity in preferences—shapes the evolution of conformity concerns. This suggests that the interaction between assortative environments and misspecified beliefs may play a broader role in the indirect evolutionary approach.

A related strand of the literature studies behavior under persistent departures from full rationality. The self-confirming equilibrium literature emphasizes epistemic consistency under limited feedback: agents may hold misspecified beliefs about the environment or about others' behavior, and such beliefs can persist whenever they are not contradicted along the realized path of play \citep{fudenberg1993self,battigalli2015self,battigalli2023learning}. In this framework, equilibrium behavior is disciplined by feedback rather than by the correctness of beliefs. To the best of our knowledge, \citet{gamba2013learning} is the only paper combining these two approaches, although in a different setting and for different preferences.

Despite addressing closely related phenomena (persistent misspecification, the robustness of non-Nash behavior, and the survival of seemingly biased decision rules), these two literatures have largely developed in isolation. Evolutionary models typically do not impose informational consistency on within-generation behavior, while self-confirming equilibrium models are generally silent about why particular misspecifications might survive evolutionary pressure. Our paper bridges these approaches by studying an environment in which agents choose actions based on self-confirming beliefs under limited feedback, while the weights they attach to different components of perceived utility---most notably social pressure---are shaped by evolutionary forces. In this sense, self-confirming equilibrium governs short-run behavior, while evolutionary selection determines which misspecified preferences are viable in the long run. This framework allows us to characterize which self-confirming behaviors are evolutionarily sustainable and the mechanisms through which distorted preferences such as conformity can emerge and persist, even when the underlying game has a unique Nash equilibrium.

More broadly, our paper relates to the literature on peer effects and social interactions in networks, which studies how behavior is shaped by peers, network structure, and assortative interaction \citep{currarini2009economic,boucher2024toward,zenou2025peer}. Relative to this literature, our focus is not on identifying reduced-form peer effects, but on providing a microfoundation for the emergence of conformity concerns when agents are imperfectly informed about their own payoff-relevant fundamentals.

Lastly, several studies show that imitation (which can be interpreted as the observable counterpart of conformity preferences) can be an effective behavioral rule. See, for instance, \citet{goerg2009experimental,schipper2009imitators,duersch2012unbeatable,licalzi2019categorization,alos2021multiple} for evidence in human settings, and \citet{laland2004social,rendell2010copy,dridi2015model} for similar findings in animal behavior. A related literature in economics studies imitation and herd behavior as the outcome of social learning and informational cascades \citep{banerjee1992simple,bikhchandani1992theory}. Our paper takes a different approach by providing a microfoundation for conformity
concerns. Rather than modeling imitation as a behavioral rule or as the outcome of a learning dynamic, we model social pressure as an endogenous preference. In our framework, imitation emerges as an equilibrium implication of agents optimally responding to perceived social incentives.

\section{Model}\label{sec:mod}

We consider a continuum of agents indexed by $i$, each choosing an effort level $x \in \mathbb{R}_+$. The environment is static. Behavior is determined by short-run optimization under subjective beliefs, while preferences are shaped by long-run evolutionary forces.

As we will see, agents' true material payoffs depend only on their own effort and on a group-specific return to effort. However, agents hold misspecified beliefs about that return and attach a subjective weight to conformity with their peers.

\paragraph{Material payoffs and subjective utility.}

Agents are partitioned into $K$ groups indexed by $k$,\footnote{%
We treat this partition as exogenous and assume that each group occupies a fixed fraction of the population. The precise group sizes play no role in the analysis below, except if one wishes to compute aggregate welfare.} where group $k$ is characterized by a true return to effort $\alpha_k \in \mathbb{R}_+$. The material payoff of an agent in group $k$ who chooses effort $x$ is
\begin{equation}\label{eq_pay_gen}
\pi_k(x,\alpha_k) = C - (\alpha_k - x)^2 .
\end{equation}
Absent any frictions, the optimal effort choice is therefore $x^*=\alpha_k$. However, agents may be misspecified about their own return to effort. Correctly specified agents know $\alpha_k$, whereas misspecified agents estimate it $\hat{\alpha}_k \neq \alpha_k$. Let $q_k \in (0,1)$ denote the fraction of agents in group $k$ who are correctly specified. The remaining fraction $1-q_k$ shares the common misspecified belief $\hat{\alpha}_k$.

Although conformity does not enter material payoffs directly, agents may behave as if aligning with peers were payoff-relevant. They may attach a subjective weight to conformity with their peers. Thus, each agent is characterized by a pair $\psi=(\alpha_{\psi,k},\lambda_\psi)$, where $\alpha_{\psi,k}=\alpha_k$ for correctly specified agents and  $\alpha_{\psi,k}=\hatalp_k$ for misspecified ones. $\lambda_\psi \in (0,1]$ measures the weight placed on conformity. Given $(\alpha_{\psi,k},\lambda_\psi)$, the subjective utility of type $\psi$ from choosing effort $x$ is
\begin{equation}\label{eq_ut/peer_k_miss}
u_\psi(x,E_k[x]) 
= (1-\lambda_\psi)\left(C - (\alpha_{\psi,k} - x)^2 \right)
+ \lambda_\psi \left( C - (x - E_k[x])^2 \right),
\end{equation}
where $E_k[x]$ denotes the expected effort of the peers faced by an agent from group $k$, and $C$ is a constant. The quadratic structure in Equation~\eqref{eq_pay_gen} is mainly useful for obtaining explicit formulas. In Appendix~\ref{secC:general_payoffs}, we consider a general strictly concave function $\pi(x,\alpha_k)=U(\alpha_k - x)$ where $U:\mathbb{R} \to \mathbb{R}$ is strictly increasing in $(-\infty,0)$ and strictly decreasing in $(0,\infty)$. We show that our main results are robust to such a specification.

\paragraph{Expectations, peer composition, and peer effects.}

Expectations $E_k[x]$ are determined by a peer composition $P=(p_{kj})_{kj}$. Let $p_{kj}$ 
denote the intensity with which an agent from group $k$ interacts with an agent 
from group $j$, with $\sum_j p_{kj}=1$. We do not impose any further restrictions on the peer composition. The matrix $P$ can be interpreted as a network of interactions, which may be directed or undirected; in particular, we do not require symmetry, so that $p_{jk}$ need not equal $p_{kj}$. These intensities capture homophily 
and biased meeting opportunities across groups. Specifically, $p_{kk}$ measures the intensity with which an agent in group $k$ interacts with another agent from group $k$, and hence the degree of assortativity in peer composition for group $k$. We assume a linear peer composition. We can conclude that the exogenous primitives of the model are the profiles $\{\alpha_k\}$, $\{p_{kj}\}$, and $\{q_k\}$, together with the constant $C$. We assume that agents know the full interaction structure, namely the profiles
\((p_{kj})_{k,j}\) and \((q_k)_k\). They may instead be misspecified about the
profiles \((\alpha_k)_k\), \((\lambda_k)_k\), and the constant \(C\).

In the evolutionary literature, such peer composition processes are typically implicit and random, and agents are assumed to understand the matching probabilities. For this reason, stable outcomes are often described as Bayesian Nash equilibria, even though agents may fail to know the true parameters of the environment.

Instead, our perspective is closer to the literature on network games with a continuum of agents \citep{galeotti2010network}. Rather than modeling stochastic matching explicitly, we interpret $E_k[x]$ as a reduced-form peer effect: the expected behavior of relevant others, shaped by homophily and interaction patterns. In this formulation, agents respond directly to the aggregate statistic $E_k[x]$, so that payoffs are deterministic conditional on expectations. 
This allows us to avoid a fully Bayesian description of the interaction. Instead, we distinguish between Nash equilibria---when agents correctly know the parameters of the environment and this is common knowledge---and self--confirming equilibria, where agents optimize given subjective beliefs that are consistent with the feedback they observe.

Let $\psi_k$ be a correctly specified type of group $k$ and $\phi_k$ a misspecified type of the same group. According to the above properties and given the primitives $\{p_{kj}\}$ and $\{q_k\}$, $E_k[x]$ is defined as follows:
\begin{equation}\label{eq_Ek[x]}
    E_k[x] = \sum_{j=1}^{K}p_{kj}(q_j\,x^{sbr}_{\psi_j} + (1-q_j) \, x^{sbr}_{\phi_j})
\end{equation}

\paragraph{Solution concepts.}

Agents may hold incorrect beliefs about the return to effort \(\alpha_k\) and may attach a subjective weight \(\lambda\) to social pressure, even though social interactions do not directly enter true material payoffs. We therefore model behavior using the concept of self--confirming equilibrium \citep{battigalli2015self}. In this framework, agents optimize given their subjective beliefs, and such beliefs need only be consistent with the payoff feedback observed along the equilibrium path.

We begin by defining the notion of subjective best response.\footnote{%
A subjective best response can be interpreted as rationalizable behavior, in the sense that it is optimal under some subjective representation of the decision problem, here given by \eqref{eq_ut/peer_k_miss}. In particular, agents correctly observe the aggregate statistic \(E_k[x]\), but may hold misspecified views about how payoffs depend on own effort and on conformity.}

\begin{definition}[Subjective best response -- SBR]\label{def:SBR}
Fix a peer composition $P=(p_{kj})_{k,j}$ and consider the true game in which material payoffs are given by $\pi_k(x,\alpha_k)$. 

Given expectations $E_k[x]$, a choice $x^{sbr}_\psi$ is a subjective best response for a type $\psi=(\alpha_{\psi,k},\lambda_\psi)$ if
\[
x^{sbr}_\psi \in \arg\max_{x\in\mathbb{R}_+} u_\psi(x,E_k[x]).
\]
\end{definition}

Suppose agents observe their realized material payoff $\pi_k(x,\alpha_k)$ together with the aggregate statistic $E_k[x]$ summarizing peers’ behavior, but do not observe the true return $\alpha_k$ directly. These observables constitute the feedback on which beliefs are evaluated. A self--confirming equilibrium is then defined as follows. 

\begin{definition}[Self-confirming equilibrium -- SCE]\label{def:SCE}
Fix a peer composition $P=(p_{kj})_{k,j}$ and consider the true game in which material payoffs are given by $\pi_k(x,\alpha_k)$.

A self-confirming equilibrium consists of a profile of effort choices $(x^{sce}_\psi)_{\psi\in\Psi}$ such that:

\begin{enumerate}
\item for every type $\psi=(\alpha_{\psi,k},\lambda_\psi)$, $x^{sce}_\psi$ is a subjective best response given expectations $E_k[x]$;
\item the beliefs $(\alpha_{\psi,k},\lambda_\psi)$ underlying $u_\psi$ are consistent with the payoff feedback generated by the equilibrium choices.
\end{enumerate}
\end{definition}

For comparison, we introduce the benchmark in which agents correctly know the parameters of the environment and maximize the true payoff function.

\begin{definition}[Nash equilibrium -- NE]\label{def_NE}
Fix a peer composition $P=(p_{kj})_{k,j}$ and consider the true game in which material payoffs are given by $\pi_k(x,\alpha_k)$.

A Nash equilibrium is a profile of effort choices $(x_\psi^{ne})_{\psi\in\Psi}$ such that for every type $\psi$
\[
x_\psi^{ne} \in \arg\max_{x\in\mathbb{R}_+} \pi_k(x,\alpha_k).
\]
\end{definition}

We now define efficiency for the social pressure.

\begin{definition}[Efficient Social Pressure Level]\label{def_lam_opt}
Fix a peer composition \(P=(p_{kj})_{k,j}\) and consider the true game in which material payoffs are given by \(\pi_k(x,\alpha_k)\).

A weight \(\lambda^*_\psi \in [0,1]\) is efficient if, given the induced subjective best response \(x^{sbr}_\psi(\lambda)\),
\[
\lambda^*_\psi \in \arg\max_{\lambda\in[0,1]} \pi_k(x^{sbr}_\psi(\lambda),\alpha_k).
\]
\end{definition}

Throughout the paper, we examine whether the efficient solution for $\lambda$ is also evolutionarily stable. To do so, we analyze the stability of $\lambda^*_\psi$ against the invasion of a type from the same group $k$ and with the same belief $\alpha_{\psi,k}$ but a different $\lambda$. Formally, we consider a population such that $\lambda_\psi = \lambda^*_\psi$ for all $\psi \in \Psi$, invaded by a small fraction $\varepsilon \in (0,1)$ of a type $\psi' = (\alpha_{\psi,k}, \lambda'_\psi)$ where $\lambda'_\psi \neq \lambda^*_\psi$. This analysis is conducted for all $\psi \in \Psi$. 
We now denote by \(x^{sbr}_{\psi}(\varepsilon)\) the subjective best response of type \(\psi\) in the presence of \(\varepsilon\) mutants of type \(\psi'\).

\begin{definition}[Evolutionarily Stable Social Pressure Level]\label{defi_lam_ess}

Fix a peer composition $P=(p_{kj})_{k,j}$ and consider the true game in which material payoffs are given by $\pi_k(x,\alpha_k)$. Consider two types $\psi = (\alpha_{\psi,k}, \lambda^*_\psi)$ and $\psi' = (\alpha_{\psi,k}, \lambda_{\psi'})$. Let $\varepsilon$ denote the frequency of type $\psi'$ in group $k$. Type $\psi$ is \emph{evolutionarily stable} against $\psi'$ if there exists $\underline{\varepsilon} > 0$ such that, for all $\varepsilon \in (0, \underline{\varepsilon}]$,
\[
\pi_k(x^{sbr}_\psi(\varepsilon), \alpha_k) \geq \pi_k(x^{sbr}_{\psi'}(\varepsilon), \alpha_k).
\]
A type $\psi$ is \emph{evolutionarily stable} if it is evolutionarily stable against all types $\psi' \neq \psi$ in $\Psi$.

\end{definition}

\section{Results}\label{sec:res}

In this section, we provide the main results of the model. We begin with a preliminary result that will be useful in what follows. 


\begin{lemma}\label{lem_l=0_nonmiss}
For any type $\psi$ satisfying $\alpha_{\psi,k}=\alpha_k$, $\lambda^*_\psi = 0$ is efficient and uniquely evolutionarily stable.
\end{lemma}

To prove the result it is sufficient to notice that since $\alpha_{\psi,k} = \alpha_k$, all $\lambda^*_\psi \neq 0$ would make type $\psi$ deviate from the payoff maximizing choice $x^* = \alpha_k$. By Lemma \ref{lem_l=0_nonmiss}, $\lambda^*_{\psi}=0$ and $x^*_\psi=\alpha_k$ for all types with $\alpha_{\psi,k}=\alpha_k$. Hence, we can simplify the notation. We denote $x^{sbr}_{\psi}=x^{sbr}_k$ and $\lambda_\psi=\lambda_k$ for all types s.t. $\alpha_{\psi,k}=\hatalp_k$.

Given Lemma~\ref{lem_l=0_nonmiss}, we can substitute $x^{sbr}_{\psi_k}$ in Equation~\eqref{eq_Ek[x]} and obtain the following expression for $E_k[x]$
\begin{equation}\label{eq_E[x]_k}
    E_k[x] = \sum_{j = 1}^K p_{kj}\left(q_j \alpha_j + (1-q_j)x^{sbr}_j \right)
\end{equation}

The subjective best response $x^{sbr}_k$ follows from the first-order condition of Equation~\eqref{eq_ut/peer_k_miss}:
\begin{equation}\label{eq_xopt_k}
    x^{sbr}_k = (1-\lambda_k) \hatalp_k + \lambda_k E_k[x]
\end{equation}

The two above equations delineate a system with $2K$ equations and $2K$ unknowns which can be reduced to a system of $K$ equations and unknowns by substituting Equation~\eqref{eq_E[x]_k} into Equation~\eqref{eq_xopt_k}. Such a system is well defined and all functions are linear. Thus, we can derive some general properties from it about the existence and uniqueness of $x^{sbr}_k$. The final system of equations we need to study takes the following form
\begin{equation}\label{eq_xopt_k_red}
    x^{sbr}_k = (1-\lambda_k) \hatalp_k + \lambda_k \sum_{j = 1}^K p_{kj}\left(q_j \alpha_j + (1-q_j)x^{sbr}_j \right)
\end{equation}

\paragraph{Existence Results.} We start our results by showing that given Equation~\eqref{eq_xopt_k_red}, a subjective best response exists and it is unique for all misspecified agents.

\begin{lemma}\label{lem_unique_x_k}
There exists a unique vector $(x^{sbr}_1,\dots,x^{sbr}_K)$ of subjective best responses.
\end{lemma}

Lemma~\ref{lem_unique_x_k}  ensures that, despite the feedback between expectations and individual choices, equilibrium behavior is well defined. In particular, misspecification and conformity concerns do not generate indeterminacy in the vector of subjective best responses.
We can use the result in Lemma~\ref{lem_unique_x_k} to prove the existence and uniqueness of $\lambda^*_\psi$ for all $\psi$. 

\begin{lemma}\label{lem_unique_laE_k[x]}
For any misspecified type, there exists a unique efficient social pressure level.
\end{lemma}

Lemma~3 establishes that, for each misspecified type, there is a unique conformity weight that maximizes true material payoffs. The proof relies on showing that optimal effort is continuous and monotone in $\lambda_k$. The intuition can be seen from Equation~\eqref{eq_E[x]_k}: increasing
\(\lambda_k\) shifts the subjective best response away from \(\hatalp_k\)
and toward the peer statistic \(E_k[x]\). Since the equilibrium system is
well behaved, this shift is monotone. This implies that the unconstrained problem admits at most one interior solution.\footnote{Note that, while $x^{sbr}_k$ is linear in $\lambda_k$, the induced payoff $\pi_k(x^{sbr}_k(\lambda_k), \alpha_k)$ is quadratic in $\lambda_k$, since deviations from $\alpha_k$ enter quadratically. Hence, the efficient peer pressure level may take interior values in $(0,1)$.} Once attention is restricted to $\lambda_k \in [0,1]$, however, the optimum may also arise at a corner, so uniqueness is not completely immediate from linearity alone.

\paragraph{Properties of interior solutions.}

We now turn to the case in which the efficient social pressure level is interior, i.e., $\lambda_k^* \in (0,1)$ satisfies Definition~\ref{def_lam_opt}. In this case, one can derive a number of additional properties. We begin by showing that such a social pressure level also satisfies Definition~\ref{defi_lam_ess}.

\begin{proposition}\label{prop_l*k_ESS}
If there exists an efficient social pressure level $\lambda^*_k \in (0,1)$, then it is also evolutionarily stable.
\end{proposition}

Proposition~\ref{prop_l*k_ESS} shows that efficiency is not only a static property of the optimal conformity weight: whenever an interior efficient social pressure level exists, it is also robust to evolutionary deviations by agents with the same misspecification but a different degree of conformity.

Having established  the evolutionary stability of $\lambda^*_k$, we can study the properties of such peer pressure parameter. Given the monotonicity of $x^{sbr}_k$ in $\lambda_k$, we show in the proof of Lemma~\ref{lem_unique_laE_k[x]} that $\lambda^*_k$ satisfies $x^{sbr}_k(\lambda^*_k) = \alpha_k$. Thus, when a $\lambda^*_k \in (0,1)$ exists, even misspecified agents will play according to their true value. From this fact, we derive a straightforward property.

\begin{corollary}\label{cor_selfconf_gen}

If $\lambda^*_k \in (0,1)$ for all misspecified agents, then the profile of choices $(x^{sbr}_k(\lambda^*_k))_k$ is both a self-confirming and a Nash equilibrium.

\end{corollary}

Corollary~\ref{cor_selfconf_gen} highlights the main implication of the model. When the efficient social pressure level is interior, social pressure exactly offsets misspecification: agents behave as if they knew the true return to effort, even though their underlying beliefs remain incorrect. In this sense, correct behavior and misspecified beliefs can coexist in equilibrium. Importantly, the profile of subjective best responses constitutes a Nash and a
self-confirming equilibrium only when all misspecified agents have an interior efficient
level of social pressure. However, even when this condition fails, some groups may still choose an interior (efficient) level of social pressure. These agents  exert the optimal level of effort, and the feedback they receive does not contradict their (misspecified) beliefs.

The result in Corollary~\ref{cor_selfconf_gen} is straightforward (so, the proof is omitted) but it can be directly used to compute the value of $\lambda^*_k$. Indeed, we know that if $\lambda^*_k \in (0,1)$, we can substitute $x^{sbr}_k$ with $\alpha_k$ in Equation~\eqref{eq_xopt_k_red}, obtaining the following equation:
$$
\alpha_k = (1-\lambda^*_k) \hatalp_k + \lambda^*_k\sum_{j = 1}^K p_{kj}\left(q_j \alpha_j + (1-q_j)x^{sbr}_j \right)
$$

Let $\Sigma_k = \sum_{j = 1}^K p_{kj}\left(q_j \alpha_j + (1-q_j)x^{sbr}_j \right)$. We can derive the following formula for the interior solution of $\lambda^*_k$
\begin{equation}\label{eq_l*_k}
    \lambda^*_k = \frac{\hatalp_k - \alpha_k}{\hatalp_k - \Sigma_k}
\end{equation}
Equation~\eqref{eq_l*_k} clarifies the economic content of the efficient social pressure level. The numerator captures the extent of misspecification relative to the true return to effort, while the denominator captures the discrepancy between the agent's perceived return and the perceived peer choices. From Equation~\eqref{eq_l*_k}, we can derive some properties of $\lambda^*_k$. We start by stating the conditions under which $\lambda^*_k$ is interior.

\begin{corollary}\label{cor_l*_k_int}
    $\lambda^*_k \in (0,1)$ if and only if one of the following two conditions holds:
    \begin{enumerate}
        \item $\hatalp_k <\alpha_k<\Sigma_k$
        \item $\Sigma_k<\alpha_k<\hatalp_k$
    \end{enumerate}
\end{corollary}

The result in the above corollary is straightforwardly derived from Equation~\eqref{eq_l*_k}, and thus, the proof is omitted. While hard to comment on, the conditions in Corollary~\ref{cor_l*_k_int} are helpful in limiting the comparative statics to interior values of $\lambda^*_k$. It should be noted that when $\lambda^*_k$ is interior for all groups $k$, since $x^{sbr}_k=\alpha_k$ for all $k$, the conditions such that $\lambda^*_k \in (0,1)$ for all $k$ boil down to $\hatalp_k <\alpha_k<\sum_{j = 1}^K p_{kj}\alpha_j$ or $\sum_{j = 1}^K p_{kj}\alpha_j<\alpha_k<\hatalp_k$. 

\paragraph{Implications for learning under limited feedback.} Corollaries~\ref{cor_selfconf_gen} and~\ref{cor_l*_k_int} imply a further insight. Suppose that agents are allowed to update their beliefs $\alpha_{\psi,k}$ through some learning process. Even then, full learning need not occur. When the conditions of Corollary~\ref{cor_l*_k_int} hold, the subjective best response coincides with the
true payoff-maximizing effort level. Hence, the feedback generated by equilibrium behavior
does not contradict agents' initial beliefs \(\alpha_{\psi,k}\). As a result, misspecification may
persist among agents with interior efficient social pressure levels.

The intuition is simple. Under an interior efficient social pressure level, agents choose the effort level $\alpha_k$, even when their beliefs about returns are incorrect. Observed outcomes are therefore consistent with those beliefs along the equilibrium path, so there is no force inducing further learning. Correct behavior can thus coexist with incorrect beliefs.

This argument applies only to interior solutions. When instead $\lambda_k^* \in \{0,1\}$, social pressure is either too weak or too strong to induce effort equal to $\alpha_k$. In such corner cases, agents' behavior may generate feedback that contradicts their beliefs, thereby creating scope for learning the true value $\alpha_k$. See Appendix~\ref{secB:res_2G} for a detailed analysis in the case $K=2$.
In those corner cases, behavior generally no longer coincides with the Nash benchmark of the true game, so the equivalence between self-confirming and Nash equilibrium established for interior solutions breaks down.

\paragraph{Comparative Statics.} In the next statement, we analyze the relation between the efficient social pressure level and two important variables, i.e., the misspecification $\hatalp_k$ and the assortativity in peer composition $p_{kk}$.

\begin{proposition}\label{prop_cstat_l*_k}
If \(\lambda_k^* \in (0,1)\), then:
\begin{enumerate}
    \item \(\lambda_k^*\) is increasing in \(\hatalp_k\) if and only if
    \(\alpha_k>\Sigma_k\).

    \item Consider a proportional increase in within-group homophily, so that
    \[
    p_{kj}(p_{kk})=(1-p_{kk})r_{kj}, \qquad j\neq k,
    \]
    with \(r_{kj}\geq 0\) and \(\sum_{j\neq k}r_{kj}=1\). Then \(\lambda_k^*\)
    is increasing in \(p_{kk}\) provided that \(\lambda_j^*\in(0,1)\) for all
    groups \(j\). Otherwise,
    \[
    \frac{\partial \lambda_k^*}{\partial p_{kk}}
    =
    \frac{(\hatalp_k-\alpha_k)(\alpha_k+g'(p_{kk}))}
    {(\hatalp_k-p_{kk}\alpha_k-g(p_{kk}))^2},
    \]
    where
    \[
    g(p_{kk})
    :=
    \sum_{j\neq k}p_{kj}(p_{kk})
    \bigl(q_j\alpha_j+(1-q_j)x_j^{sbr}(p_{kk})\bigr).
    \]

    \item The effect of \(\{q_j\}\) is generally ambiguous. If there is at least  one group $\ell$ (including $\ell=k$) such that 
    \(x_\ell^{sbr}=\alpha_\ell\), then
    \[
    \frac{\partial \lambda_k^*}{\partial q_{\ell}}=0
    .
    \]
    Otherwise,
    \[
    \frac{\partial \lambda_k^*}{\partial q_j}
    =
    \frac{\hatalp_k-\alpha_k}{(\hatalp_k-\Sigma_k)^2}
    \frac{\partial \Sigma_k}{\partial q_j},
    \]
    where \(\partial \Sigma_k/\partial q_j\) includes both the direct effect on
    group \(j\)'s average behavior and the indirect equilibrium effects on the
    subjective best responses of other groups.
\end{enumerate}
\end{proposition}

The first result in Proposition~\ref{prop_cstat_l*_k} states that the relationship between $\lambda^*_k$ and $\hatalp_k$ (i.e., the degree of misspecification) depends on the relationship between $\alpha_k$ and $\Sigma_k$. If $\alpha_k > \Sigma_k$, then an increase in $\hatalp_k$ also determines an increase in the optimal degree of peer pressure. The opposite happens when $\alpha_k < \Sigma_k$. Interestingly, the result does not depend on the degree of the bias of misspecified agents in group $k$. 

The second result states that, for a group \(k\), peer pressure is increasing in within-group
assortativity, under a proportional reduction of out-group peer weights, provided that all
agents exert the optimal level of effort. This is intuitive: if every agent selects the optimal level of effort, the higher the level of homophily, the more intensely agents interact with peers who have
the same return to effort. As a consequence, conforming to peers becomes more attractive, since peer behavior is more likely to guide agents toward the optimal effort level \(\alpha_k\). If some agents do not exert effort efficiently, their behavior may also be influenced by the assortativity of group \(k\). In this case, changes in \(p_{kk}\) affect not only the composition of group \(k\)’s peers, but also the behavior of other groups through equilibrium interactions. As a result, the effect of \(p_{kk}\) depends on whether it moves \(\Sigma_k\) closer to \(\alpha_k\). If so, an increase in assortativity raises the efficient level of peer pressure. Otherwise, it may reduce it. For instance, a higher \(p_{kk}\) may increase the likelihood of interacting with agents whose effort is far from \(\alpha_k\), thereby lowering \(\lambda_k^*\).

An illustrative example with $K=2$ of the dependence on $p_{kk}$ when $x^{sbr}_k=\alpha_k$ for all $k$ can be found in Figure~\ref{fig:lj_bounded_K2}. Note that for $K=2$, only two weights matter in $E_k[x]$ as there are only two groups. Thus, we can simplify the peer composition process calling $p$ the assortativity in peer composition (a more detailed discussion on the case with $K=2$ can be found in Appendix~\ref{secB:res_2G}). In the figure, $\lambda^*_k >0$ for all $k \in \{1,2\}$ but reaches $1$ as $p$ goes to $1$ (similarly to what Proposition~\ref{prop_cstat_l*_k} predicts).

\begin{figure}[ht!]
    \centering
    \includegraphics[width=0.7\linewidth]{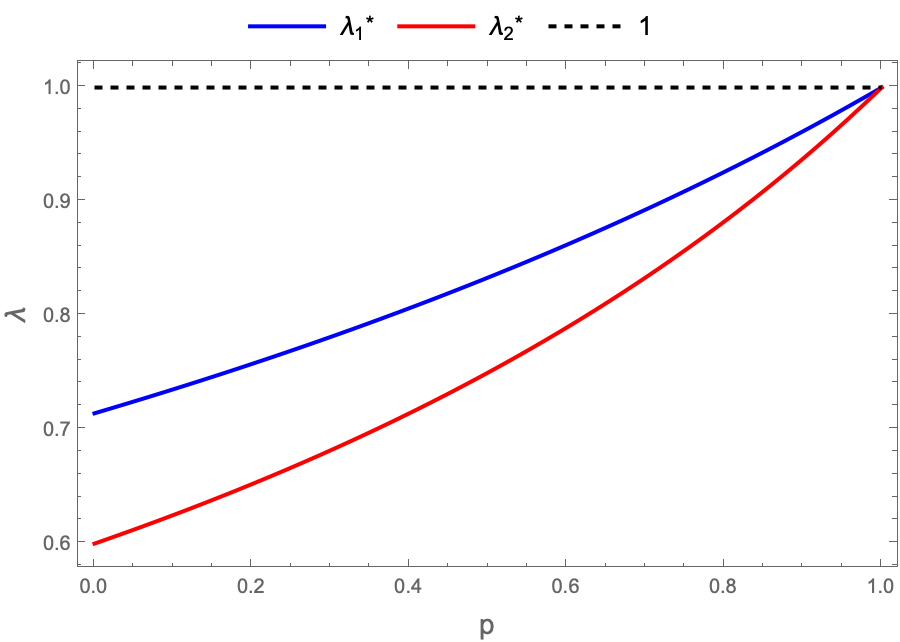}
    \caption{This graph depicts $\lambda^*_1$ and $\lambda^*_2$ as a function of $p$ fixing $\alpha_1 = 10$, $\alpha_2=12$, $\hat \alpha_1 = 5$, $\hat \alpha_2=15$.}
    \label{fig:lj_bounded_K2}
\end{figure}

The third result concerns the effect of the fraction of correctly specified agents on the efficient level of peer pressure. Notably, an increase in the likelihood of meeting a correctly specified agent from group $k$ does not affect the pressure experienced by group $k$ itself when $\lambda^*_k \in (0,1)$. This occurs because, in equilibrium, all agents in group $k$ exert effort equal to $\alpha_k$. The same result applies for all groups where also misspecified agents play the optimal level of effort.

By contrast, the effect of an increase in the likelihood of meeting a correctly specified agent from another group depends on how this change affects that group's behavior. If an increase in $q_j$ for some group $j$ such that $x^{sbr}_j \neq \alpha_j$ moves its subjective best response closer to $\alpha_k$, then it may also increase $\lambda^*_k$.

\section{The value and costs of social information}
\label{sec:policy}

We now interpret the equilibrium value of social monitoring as a reduced-form measure of the value that agents attach to peer information. This does not introduce a separate market for information into the model, but it allows us to discuss how an intermediary with access to peer data could extract informational rents.
The analysis in this section should be read as an interpretation of the equilibrium value of peer information within the model, rather than as a fully specified model of information sales by an intermediary.

In the evolutionarily stable self-confirming equilibrium with interior efficient peer pressure, agents choose effort equal to their true return $\alpha_k$. Peer pressure and misperceptions therefore need not generate allocative distortions in equilibrium: behavior coincides with the efficient benchmark.

Nevertheless, agents face a trade-off. While conditioning behavior on the social statistic $E_k[x]$ requires observing aggregate behavior and solving a fixed point problem, the autarchic rule $x=\hat\alpha_k$ ignores social information and is cognitively simpler. Although we do not explicitly model such computational or attention costs, the comparison highlights that monitoring entails a strictly more demanding decision problem.

Our analysis characterizes stationary outcomes. Outside these outcomes, misspecified agents need not best respond to the prevailing environment. In that case, conditioning behavior on aggregate peer outcomes may expose agents to temporary utility losses when the relevant social statistic changes. These costs may be larger when peer concerns $\lambda_k$ are stronger, since behavior then responds more sharply to variation in $E_k[x]$. We do not model such adjustment dynamics explicitly, and view this only as a possible interpretation of why social monitoring may be costly.

\paragraph{Perceived value of monitoring.}

Recall that a misspecified type $\psi=(\hat\alpha_k,\lambda_k)$ chooses effort to maximize the utility in Equation \eqref{eq_ut/peer_k_miss}. If the agent does not observe $E_k[x]$, a natural benchmark is the autarchic rule
$x=\hat\alpha_k$, which minimizes the individual loss term and ignores the social component. 
When $E_k[x]$ is observed, the subjective best response ($x^{sbr}_k$) is given by Equation \eqref{eq_xopt_k}.

The utility difference between these two choices can be computed from Equation \eqref{eq_ut/peer_k_miss}, and is
\[
u_\psi(x^{sbr}_k,E_k[x]) - u_\psi(\hat\alpha_k,E_k[x])
= \lambda_k^2 (E_k[x]-\hat\alpha_k)^2.
\]

We define
\begin{equation}\label{eq_Deltak}
\Delta_k = \lambda^2_k (E_k[x]-\hat\alpha_k)^2   
\end{equation}

as the \emph{perceived value of monitoring peer behavior}. 
This quantity captures the subjective gain from conditioning behavior on the observed social environment rather than following the simpler autarchic rule.

It is increasing and convex in \(\lambda_k\) and in the absolute perceived discrepancy
\(|E_k[x]-\hat{\alpha}_k|\). As a result, marginal increases in either peer concerns $\lambda$ or in the perceived discrepancy $|E_k[x]-\hat\alpha_k|$ raise the perceived value of monitoring more strongly when the other component is already large.

The perceived importance of social information is therefore self--reinforcing: larger conformity concerns and larger perceived discrepancies interact multiplicatively in determining this perceived value.
The key point is that social monitoring has value only because misspecified agents use peer behavior as a device to offset their distorted self-assessment. In this sense, the willingness to pay for peer information is endogenous to misspecification.

\paragraph{Informational rents.}

Suppose that observing $E_k[x]$ requires access to an informational intermediary. 
A discriminatory monopolist selling such information can extract, from each type $k$, up to the perceived surplus $\Delta_k$.




Proposition~\ref{prop_cstat_l*_k} suggests how the evolutionarily stable level of peer
concerns responds to the informational environment away from the interior Nash benchmark.
In particular, changes in within-group peer composition \(p_{kk}\) may affect both the
efficient level of peer pressure and the perceived discrepancy
\(|\hat{\alpha}_k-E_k[x]|\). The sign of these effects is generally ambiguous, because changes
in peer composition may move expected peer behavior either closer to or farther from the
true return \(\alpha_k\). Since \(\Delta_k\) is increasing in both \(\lambda_k\) and
\(|E_k[x]-\hat{\alpha}_k|\), such changes may either increase or decrease the scope for
informational rent extraction.


Depending on the effect of homophily, the mechanism may either reinforce peer concerns
or generate non-monotonic dynamics. Two scenarios can emerge. In the first, more segregated societies, either structurally, through higher within-group homophily $p_{kk}$, or in terms of perceived differences across types, may select stronger peer concerns in equilibrium. In the second, greater structural segregation may weaken peer effects, while segregation in beliefs may simultaneously increase the value of social information.


In the first scenario, the interplay between two forms of segregation—polarization, captured by the discrepancy between $\hat\alpha_k$ and $E_k[x]$, and homophily, captured by $p_{kk}$—generates a self-reinforcing cycle. A larger discrepancy increases the value of conformity, while greater homophily makes peer behavior more influential, leading to higher peer pressure. In turn, stronger peer concerns raise the perceived value of monitoring and expand the scope for informational rent extraction.

In the second scenario, the interplay between polarization and homophily generates opposing forces. Greater homophily reduces peer pressure, while a larger discrepancy between $\hatalp_k$ and $E_k[x]$ increases the value of conformity. As a result, the scope for informational rents depends on the trade-off between these two forces: the stronger the homophily effect, the lower the informational rents, whereas the stronger the discrepancy effect, the higher the informational rents.


\paragraph{Informational rents under interior equilibrium.}

We now analyze informational rents in the interior equilibrium, where agents select the efficient social pressure level and the subjective best response satisfies $x_k^{sbr}=\alpha_k$.

\begin{proposition}\label{prop:info_sce}
For each group, \(k\) under the efficient social pressure level,
\[
\Delta_k=(\hat{\alpha}_k-\alpha_k)^2.
\]
Hence, conditional on the induced effort level being optimal, the perceived value of monitoring
depends only on the degree of misspecification.
\end{proposition}


Proposition~\ref{prop:info_sce} identifies a sharp benchmark: conditional on the induced subjective-best-response being optimal, larger gaps between perceived and true returns generate a higher perceived value of social information and expand the scope for informational rent extraction by intermediaries.

It is important to stress that the mechanism identified in the model does not, in itself, generate permanent allocative inefficiency. In the evolutionarily stable self-confirming equilibrium, agents choose effort equal to their true return $\alpha_k$, so behavior coincides with the efficient benchmark. The social network is taken as given, and the long-run allocation does not hinge on the initial misspecifications.
The relevant costs arise instead on the informational side. Agents condition their choices on perceived fundamentals $\hat\alpha_k$ and on the observed social statistic $E_k[x]$, so the value of monitoring reflects the interaction between a given network structure and misspecified beliefs about types. In this sense, informational intermediaries need not distort fundamentals or the network directly; rather, they may extract value by providing access to a social signal that agents perceive as useful.

This observation suggests a broader interpretation of the mechanism. In environments where access to peer behavior is mediated by platforms, dashboards, or disclosure rules, the design of the informational environment may affect the perceived value of social information. For instance, institutions that selectively display peer outcomes may shape how informative others' behavior appears to be, and therefore how strongly agents rely on it in their decisions. More generally, environments with greater misspecification create more scope for social information to be valuable and, in turn, for intermediaries to extract informational rents.
These considerations should be interpreted as implications of the mechanism rather than as the output of a formal model of platform design or regulation. The main point is narrower: even when conformity does not create long-run distortions in behavior, the demand for social information can still have economic value, and this value is higher in environments where agents are more uncertain about their own payoff-relevant environment. 

In the next section, we discuss the empirical relevance of the model and how its mechanisms map into observed patterns of two streams of the literature.

\section{Interpretation and empirical relevance}\label{sec:disc}

The model suggests that conformity concerns should be more likely to emerge when agents are uncertain about their own returns and when peer behavior is informative about how to correct such misperceptions.
This section discusses two empirical environments in which the mechanism highlighted by the model may be relevant, and draws out a set of qualitative implications rather than sharp causal predictions.
These examples should be interpreted as illustrative domains of application, rather than as direct tests of the model.

\paragraph{Peer pressure in adolescence.} According to findings in the cognitive and developmental psychology literature, adolescence is the period in which individuals show heightened sensitivity to peer influence and social evaluation \citep{blakemore2014adolescence,knoll2017age}. Neuro-scientific accounts suggest that this pattern may reflect elevated neuro-plasticity during adolescence, together with ongoing structural and functional reorganization of social and affective brain networks. In this phase, cortical pruning and increasing white-matter connectivity coincide with heightened responsivity in reward-related systems, potentially rendering social signals particularly salient. Consistent with this interpretation, adolescents appear more susceptible than adults to peer effects in domains such as risk-taking behavior \citep{viner2012adolescence,knoll2015social}, prosocial behavior \citep{chierchia2020prosocial}, and affective responses to social exclusion \citep{sebastian2010social}. 

A growing body of work suggests that susceptibility to social influence increases when individuals face greater uncertainty about the outcomes of their actions. Adolescents may be more uncertain about these outcomes and therefore rely more strongly on social information when making decisions \citep[][]{ciranka2020bayesian,osmont2021peers,ciranka2025internal}. Moreover, adolescents may hold uncertain or unstable representations of their own preferences and, as a result, place greater weight on social cues. Building on this perspective, recent computational and longitudinal studies provide evidence that adolescents’ susceptibility to peer influence is partly driven by uncertainty about their own preferences, and that declining preference uncertainty predicts a reduction in social influence over time \citep[][]{moutoussis2016people,reiter2021preference}.
Our model is consistent with these findings and offers one possible mechanism behind them. In particular, agents in our model are more susceptible to peer influence the greater their degree of misspecification (Proposition~\ref{prop_cstat_l*_k}), that is, the greater their uncertainty about their own preferences. Moreover, as shown in Section~\ref{sec:policy}, the perceived value of social information is increasing in the degree of misspecification.

One possible interpretation is developmental. If uncertainty about $\alpha_{\psi,k}$ declines over time, individuals' perceived preferences $\hat{\alpha}_k$ would move closer to their true parameters $\alpha_k$. As this happens, both the strength of social pressure and the perceived importance of the network decline. This interpretation is consistent with empirical evidence suggesting that susceptibility to peer influence is stronger during adolescence and attenuates as preferences become more stable over time \citep{blakemore2014adolescence,moutoussis2016people,reiter2021preference}. In this sense, the model suggests a mechanism through which susceptibility to peer pressure may decrease with age, insofar as uncertainty about one's own preferences is progressively reduced.

\paragraph{Cascade effects in consumption.} A related empirical regularity concerns the timing of cascade effects in product markets. The social learning literature emphasizes that the actions of early adopters can strongly influence subsequent decisions, generating path-dependent information cascades \citep{bikhchandani2024information}. Such mechanisms are particularly relevant in environments where agents have limited prior information, as in the introduction of new products. Empirical studies in marketing and innovation similarly show that early diffusion stages are characterized by stronger peer influence and word-of-mouth amplification \citep{golder2004growing}. In these settings, the behavior of early adopters may provide a signal to later consumers about the likely attractiveness of the product.

This interpretation is related to the mechanism in our model. There, the behavior of better-informed agents can affect the choices of misspecified agents through social influence. In this sense, social signals may help individuals choose actions that are closer to those implied by their true payoff-relevant characteristics.

Using micro-level data, \citet{feldman2019social} show that reviews from early buyers influence subsequent purchases by new consumers. Social information therefore reduces the uncertainty faced by potential buyers and increases the likelihood of adoption in the early phases of product diffusion. Consistent with this mechanism, \citet{bailey2022peer} show that peer effects are particularly pronounced for newly introduced goods and attenuate as consumers accumulate experience with a specific product. These findings are consistent with the role of social information in our model. As shown in Section~\ref{sec:policy}, the perceived value of observing peers' behavior increases with the degree of misspecification of agents' own preferences. In environments where individuals are uncertain about their tastes, social signals therefore become particularly valuable for decision-making.

While much of this literature emphasizes uncertainty about product quality, a complementary interpretation is that consumers may also be uncertain about their own tastes when encountering novel goods. In such environments, observing others' choices may provide not only information about product characteristics, but also cues about the likely desirability of the product for oneself. Our model offers a complementary theoretical lens for thinking about this possibility. When preferences are misspecified, social signals have value because they help agents make decisions that are closer to those implied by their true payoff-relevant parameters. As experience accumulates and misspecification declines, the value of social information should decline as well. In this sense, the model suggests one mechanism through which social influence may be especially strong in the early stages of product diffusion, when both product characteristics and individual tastes are less well understood.

More generally, the model suggests that the value of social information need not coincide across market participants. Consumers may value social signals because these help them reduce uncertainty about their own payoff-relevant environment, whereas firms may value the same signals primarily through their effects on demand.

\section{Conclusions}\label{sec:concl}

We study a model in which conformity preferences may emerge as an efficient response to misspecified beliefs about one's own return to effort, even though social interactions do not directly enter material payoffs. We show that both the subjective best response and the efficient level of peer pressure are uniquely defined. For correctly specified agents, the efficient level of peer pressure is zero, whereas for misspecified ones it may be interior. In the latter case, it induces agents to choose effort equal to their true return and thus aligns behavior with the efficient benchmark.

A central implication of the model is that, when the efficient level of peer pressure is interior
for all misspecified agents, the induced subjective-best-response profile is both self-confirming
and Nash. Agents may therefore behave optimally without learning the true fundamentals: misspecified beliefs persist because the feedback generated along the equilibrium path does not contradict them. The intensity of peer pressure is in turn shaped by the degree of misspecification and by the structure of social interactions, 
with more homogeneous interactions leading to stronger conformity incentives when peers
behave optimally.

Finally, the model implies that the informational environment is not neutral. Although peer pressure does not generate long-run allocative distortions, observing others' behavior may still be valuable from the agents' subjective perspective and may therefore support informational rents. Conditional on the induced profile being Nash, the perceived value of social information
depends only on the degree of misspecification. As a result, environments in which agents are more uncertain about their own payoff-relevant environment also generate greater scope for informational rent extraction. More broadly, the paper provides a theoretical framework for understanding when conformity concerns may emerge, persist, and become stronger.


\setlength{\bibhang}{0pt}
\bibliographystyle{apalike}
\bibliography{biblio.bib}


\begin{appendices}

\renewcommand{\thetable}{A\arabic{table}}

\renewcommand{\thefigure}{A\arabic{figure}}

\renewcommand{\thelemma}{A.\arabic{lemma}}

\renewcommand{\thecorollary}{A.\arabic{corollary}}

\renewcommand{\theequation}{A\arabic{equation}}

\section{Proofs}\label{secA_proofs}

\begin{proof}[Proof of Lemma~\ref{lem_l=0_nonmiss}]

\text{ }

We start from the first-order condition in Equation~\eqref{eq_xopt_k} and note that, for a type \(\psi\) such that \(\alpha_{\psi,k}=\alpha_k\), we have
\[
x^{sbr}_\psi=(1-\lambda_\psi)\alpha_k+\lambda_\psi E_k[x].
\]

Whatever the value of \(E_k[x]\), substituting into Equation~\eqref{eq_pay_gen} gives
\[
\pi_k(x^{sbr}_\psi,\alpha_k)
=
C-\left(\alpha_k-\bigl((1-\lambda_\psi)\alpha_k+\lambda_\psi E_k[x]\bigr)\right)^2.
\]

Thus, for all values of \(E_k[x]\neq \alpha_k\), \(\lambda_\psi^*=0\) maximizes the above expression.

Evolutionary stability follows directly. Consider a small fraction \(\varepsilon\in(0,1)\) of a mutant type \(\psi'=(\alpha_k,\lambda_{\psi'})\), with \(\lambda_{\psi'}\neq 0\), invading the population. Let \(\psi=(\alpha_k,0)\) denote the incumbent type with correctly specified beliefs and no peer pressure. Then
\[
x^{sbr}_{\psi}(\varepsilon)=\alpha_k
\qquad \text{for all } \varepsilon\in(0,1),
\]
and therefore
\[
\pi_k(x^{sbr}_{\psi}(\varepsilon),\alpha_k)=C
\qquad \text{for all } \varepsilon\in(0,1).
\]
By contrast, any mutant with \(\lambda_{\psi'}\neq 0\) chooses
\[
x^{sbr}_{\psi'}(\varepsilon)=(1-\lambda_{\psi'})\alpha_k+\lambda_{\psi'}E_k[x],
\]
which differs from \(\alpha_k\) whenever \(E_k[x]\neq \alpha_k\). Hence
\[
\pi_k(x^{sbr}_{\psi}(\varepsilon),\alpha_k)
\geq
\pi_k(x^{sbr}_{\psi'}(\varepsilon),\alpha_k)
\qquad
\text{for all } \lambda_{\psi'}\in(0,1),\; \varepsilon\in(0,1).
\]
Therefore, type \(\psi\) satisfies Definition~\ref{defi_lam_ess}.
\end{proof}

\begin{proof}[Proof of Lemma~\ref{lem_unique_x_k}]
\text{ }

In order to prove that the system in Equation~\eqref{eq_xopt_k_red} is linear in $x^{sbr}_k$, we need to isolate $x^{sbr}_k$ in the right-hand term of the above equation
$$
x^{sbr}_k = (1-\lambda_k) \hatalp_k + \lambda_k 
\left(\sum_{j \neq k} p_{kj}\left( q_j \alpha_j + (1-q_j)x^{sbr}_j \right) + p_{kk}\left( q_k \alpha_k + (1-q_k)x^{sbr}_k \right) \right)
$$

We can now bring to the left-hand side all terms involving an unknown and obtain the following system
$$
x^{sbr}_k -\lambda_k 
\left(\sum_{j \neq k} p_{kj}\left( q_j \alpha_j + (1-q_j)x^{sbr}_j \right) + p_{kk}\left( q_k \alpha_k + (1-q_k)x^{sbr}_k \right) \right) = (1-\lambda_k) \hatalp_k
$$
$$
x^{sbr}_k(1 -p_{kk}\lambda_k (1-q_k)) -\lambda_k 
\left(\sum_{j \neq k} p_{kj}\left( q_j \alpha_j + (1-q_j)x^{sbr}_j \right) + p_{kk}\left( q_k \alpha_k \right) \right) = (1-\lambda_k) \hatalp_k
$$

which, ultimately, can be reworked as follows:
$$
x^{sbr}_k(1 -p_{kk}\lambda_k (1-q_k)) -\lambda_k 
\sum_{j \neq k} p_{kj} (1-q_j)x^{sbr}_j  = (1-\lambda_k) \hatalp_k + \lambda_k\sum_{j = 1}^K p_{kj}q_j\alpha_j
$$

The above equation defines a system with $K$ linear equations and $K$ unknowns. Isolating the coefficients on the left-hand side of the above equation, we can find the coefficient matrix of such system. The diagonal elements of such a matrix take the form
$$
a_{kk} = 1 -p_{kk}\lambda_k (1-q_k)
$$

while the off-diagonal elements are
$$
a_{kj} = -\lambda_k p_{kj} (1-q_j) \text{ for } j \neq k
$$





From standard results in linear algebra, the system has a unique solution if the
coefficient matrix has full rank. To show this, it is enough to show that the
coefficient matrix is strictly diagonally dominant. Since
\[
a_{kk}=1-p_{kk}\lambda_k(1-q_k)>0
\]
and
\[
a_{kj}=-\lambda_k p_{kj}(1-q_j)\leq 0
\qquad \text{for } j\neq k,
\]
strict diagonal dominance requires
\[
|a_{kk}|>\sum_{j\neq k}|a_{kj}|,
\]
that is,
\[
1-p_{kk}\lambda_k(1-q_k)
>
\lambda_k\sum_{j\neq k}p_{kj}(1-q_j).
\]
Equivalently,
\[
1>
\lambda_k\left[
p_{kk}(1-q_k)+\sum_{j\neq k}p_{kj}(1-q_j)
\right]
=
\lambda_k\sum_{j=1}^K p_{kj}(1-q_j).
\]
Since \(\lambda_k\leq 1\), \(\sum_{j=1}^Kp_{kj}=1\), and \(q_j\in(0,1)\) for all \(j\),
\[
\lambda_k\sum_{j=1}^K p_{kj}(1-q_j)<1.
\]
Hence the coefficient matrix is strictly diagonally dominant and therefore has full rank.
The system therefore admits a unique solution.

Moreover, the coefficient matrix has positive diagonal entries and non-positive off-diagonal
entries. Since it is strictly diagonally dominant, it is a nonsingular \(M\)-matrix. Hence its
inverse is entrywise non-negative.
\end{proof}

\begin{proof}[Proof of Lemma~\ref{lem_unique_laE_k[x]}]
\text{ }

We start by analyzing the maximization problem in~\eqref{eq_pay_gen}:
\[
\max_{\lambda_k\in[0,1]} \; C-\left(\alpha_k-x^{sbr}_k(\lambda_k)\right)^2 .
\]

Note that each \(x^{sbr}_k\) depends on all parameters of the model, namely \(q_j\), \(p_{kj}\), and \(\lambda_j\) for all \(j\), but for conciseness we make explicit only the dependence on \(\lambda_k\). The first-order condition for an interior solution is
\[
2\bigl(\alpha_k-x^{sbr}_k(\lambda_k)\bigr)\frac{\partial x^{sbr}_k(\lambda_k)}{\partial \lambda_k}=0.
\]

If \(x^{sbr}_k(\lambda_k)\) is monotone in \(\lambda_k\), then the only possible interior solution to the maximization problem is given by
\[
\alpha_k=x^{sbr}_k(\lambda_k).
\]

To establish this result, it is sufficient to study Equation~\eqref{eq_xopt_k}. There, $\lambda_k$ appears linearly, weighting $\hat{\alpha}_k$ and $E_k[x]$. While $\hatalp_k$ is exogenously given, the term $E_k[x]$ depends on $\lambda_k$. Specifically, $E_k[x]$ is a weighted average with non-negative weights of terms that depend on $\lambda_k$ as well (given the system of equations in \eqref{eq_xopt_k_red}). We already know from Lemma~\ref{lem_unique_x_k} that such system is well behaved, and admits a unique solution. 
The coefficient matrix associated with the system has positive diagonal entries,
non-positive off-diagonal entries, and is strictly diagonally dominant, as shown in
the proof of Lemma~\ref{lem_unique_x_k}. Hence it is a nonsingular \(M\)-matrix, and
its inverse is entrywise non-negative. Therefore, changes in the right-hand side of
the system propagate through the vector of subjective best responses without reversing
signs.
As a result, the direct shift induced by changing \(\lambda_k\) propagates through the
system without sign reversals. Hence, \(x^{sbr}_k(\lambda_k)\) is monotone in
\(\lambda_k\).

Therefore, solving the above maximization problem reduces to solving
\[
\alpha_k=x^{sbr}_k(\lambda_k).
\]
Since \(x^{sbr}_k(\lambda_k)\) is continuous and monotone in \(\lambda_k\), this equation has at most one solution. Hence, the maximization problem admits at most one interior solution. If no interior solution exists, then the optimum must be attained at a corner, namely at \(\lambda_k=0\) or \(\lambda_k=1\).
\end{proof}

\newpage

\begin{proof}[Proof of Proposition \ref{prop_l*k_ESS}]

\text{ }

To start the proof, note that when an interior efficient \(\lambda_k^*\) exists, each misspecified agent of group \(k\) with trait \((\hat{\alpha}_k,\lambda_k^*)\) maximizes the fitness function. Thus, to prove the stability of \(\lambda_k^*\), it is sufficient to show that if a small fraction \(\varepsilon\) of agents with \(\alpha_{\psi,k}=\hat{\alpha}_k\) and \(\lambda_k\neq\lambda_k^*\) invades, then incumbents still obtain strictly higher fitness than mutants.

Consider the system defined in Equation~\eqref{eq_xopt_k_red}. This system is continuous in the population shares. Now consider an invasion of mutants of type
\[
\psi'=(\hat{\alpha}_k,\lambda_k'), \qquad \lambda_k'\neq\lambda_k^*.
\]
That is, a fraction \(\varepsilon\) of misspecified agents in group \(k\) assigns a weight different from the efficient one. The implicit equation for the incumbent subjective best response in group \(k\) becomes
{\small 
\begin{equation*}
x^{sbr}_k(\varepsilon)
=
(1-\lambda_k^*)\hat{\alpha}_k
+
\lambda_k^*
\left[
p_{kk}\Bigl(q_k\alpha_k+(1-q_k)\bigl(\varepsilon x^{sbr}_{\psi'}(\varepsilon)+(1-\varepsilon)x^{sbr}_k(\varepsilon)\bigr)\Bigr)
+
\sum_{j\neq k} p_{kj}\bigl(q_j\alpha_j+(1-q_j)x^{sbr}_j(\varepsilon)\bigr)
\right].
\end{equation*}}

Hence \(x^{sbr}_k(\varepsilon)\) is continuous in \(\varepsilon\). By construction,
\[
x^{sbr}_k(0)=\alpha_k.
\]
Therefore,
\begin{equation}\label{eq_x*k_wm_revised}
\pi(x^{sbr}_k(\varepsilon),\alpha_k)
=
C-\bigl(\alpha_k-x^{sbr}_k(\varepsilon)\bigr)^2,
\end{equation}
and continuity implies
\[
\lim_{\varepsilon\to 0}\bigl(\alpha_k-x^{sbr}_k(\varepsilon)\bigr)^2=0.
\]

Now consider the mutant type \(\psi'\). Since the efficient social pressure level is unique and \(\lambda_k'\neq\lambda_k^*\), the mutant
does not choose the materially optimal action at \(\varepsilon=0\). Hence
\(x_{\psi'}^{sbr}(0)\neq \alpha_k\). 

Define
\[
\Delta:=\bigl(\alpha_k-x^{sbr}_{\psi'}(0)\bigr)^2>0.
\]
Again by continuity,
\begin{equation}\label{eq_x*mut_wm_revised}
\pi(x^{sbr}_{\psi'}(\varepsilon),\alpha_k)
=
C-\bigl(\alpha_k-x^{sbr}_{\psi'}(\varepsilon)\bigr)^2,
\end{equation}
with
\[
\lim_{\varepsilon\to 0}\bigl(\alpha_k-x^{sbr}_{\psi'}(\varepsilon)\bigr)^2=\Delta>0.
\]

Comparing \eqref{eq_x*k_wm_revised} and \eqref{eq_x*mut_wm_revised}, it follows that there exists \(\bar\varepsilon>0\) such that, for every \(\varepsilon\in(0,\bar\varepsilon)\),
\[
\pi(x^{sbr}_k(\varepsilon),\alpha_k)
>
\pi(x^{sbr}_{\psi'}(\varepsilon),\alpha_k).
\]
Therefore, after a sufficiently small invasion, incumbents with trait \(\lambda_k^*\) obtain strictly higher fitness than mutants with \(\lambda_k'\neq\lambda_k^*\). Hence \(\lambda_k^*\) satisfies Definition~\ref{defi_lam_ess}.
\end{proof}

\begin{proof}[Proof of Proposition~\ref{prop_cstat_l*_k}]

\text{ }

Let us start by studying the SBR of groups $k$ such that $x^{sbr}_k \neq \alpha_k$. Let $K_1$ be the partition of the population including groups such that $x^{sbr}_k=\alpha_k$ and $K_2$ be the partition of the population including groups such that $x^{sbr}_k \neq \alpha_k$. Let us call $x^{sbr}_{k2}$ the subjective best response of the second group. We know that $x^{sbr}_{k2}$ takes the form
\begin{equation}\label{eq_sbr_notne}
    x^{sbr}_{k_2}
=
(1-\lambda_{k_2})\hat{\alpha}_{k_2}
+\lambda_{k_2}
\left(
\sum_{j\in K_1} p_{kj}\alpha_j
+
\sum_{\ell\in K_2} p_{k\ell}
\bigl(q_\ell\alpha_\ell+(1-q_\ell)x_\ell^{sbr}\bigr)
\right).
\end{equation}

From this equation, we can derive some useful properties for studying the comparative
statics of \(\lambda_k^*\). Throughout, let
\[
\bar{x}_j:=q_j\alpha_j+(1-q_j)x_j^{sbr},
\qquad
\Sigma_k:=\sum_{j=1}^K p_{kj}\bar{x}_j .
\]
Thus, \(\Sigma_k\) is the average effort expected by agents in group \(k\).

We begin by studying the derivative with respect to \(\hatalp_k\). From
Equation~\eqref{eq_l*_k},
\[
\lambda_k^*=\frac{\hatalp_k-\alpha_k}{\hatalp_k-\Sigma_k}.
\]
Taking the direct derivative with respect to \(\hatalp_k\), holding the induced peer
statistic \(\Sigma_k\) fixed, gives
\begin{align*}
\frac{\partial \lambda_k^*}{\partial \hatalp_k}
&=
\frac{(\hatalp_k-\Sigma_k)-(\hatalp_k-\alpha_k)}
     {(\hatalp_k-\Sigma_k)^2}  \\
&=
\frac{\alpha_k-\Sigma_k}{(\hatalp_k-\Sigma_k)^2}.
\end{align*}
Since the denominator is strictly positive, the sign of the derivative is determined by
\(\alpha_k-\Sigma_k\). Hence
\[
\frac{\partial \lambda_k^*}{\partial \hatalp_k}>0
\quad\Longleftrightarrow\quad
\alpha_k>\Sigma_k .
\]

We now study the relationship between \(\lambda_k^*\) and \(p_{kk}\). Since the rows of
the peer-composition matrix sum to one, a change in \(p_{kk}\) must be accompanied by
a change in the remaining weights in row \(k\). We consider the natural proportional
variation
\[
p_{kj}(p_{kk})=(1-p_{kk})r_{kj}, \qquad j\neq k,
\]
where \(r_{kj}\geq 0\) and \(\sum_{j\neq k}r_{kj}=1\). This keeps the relative composition
of out-group peers fixed while varying the degree of within-group homophily.

Suppose first that \(\lambda_j^*\in(0,1)\) for all groups \(j\). Then
\(x_j^{sbr}=\alpha_j\) for all \(j\), and therefore
\[
\Sigma_k=\sum_{j=1}^K p_{kj}\alpha_j
=
p_{kk}\alpha_k+\sum_{j\neq k}p_{kj}\alpha_j .
\]
Let
\[
f(p_{kk}):=\sum_{j\neq k}p_{kj}(p_{kk})\alpha_j
=(1-p_{kk})A,
\qquad
A:=\sum_{j\neq k}r_{kj}\alpha_j .
\]
Then
\[
f'(p_{kk})=-A=-\frac{f(p_{kk})}{1-p_{kk}}.
\]
Using Equation~\eqref{eq_l*_k}, we can write
\[
\lambda_k^*
=
\frac{\hatalp_k-\alpha_k}
     {\hatalp_k-p_{kk}\alpha_k-f(p_{kk})}.
\]
Differentiating with respect to \(p_{kk}\) yields
\begin{align*}
\frac{\partial \lambda_k^*}{\partial p_{kk}}
&=
\frac{-(\hatalp_k-\alpha_k)(-\alpha_k-f'(p_{kk}))}
     {\left(\hatalp_k-p_{kk}\alpha_k-f(p_{kk})\right)^2} \\
&=
(\hatalp_k-\alpha_k)
\frac{\alpha_k+f'(p_{kk})}
     {\left(\hatalp_k-p_{kk}\alpha_k-f(p_{kk})\right)^2}  \\
&=
(\hatalp_k-\alpha_k)
\frac{\alpha_k-A}
     {\left(\hatalp_k-p_{kk}\alpha_k-f(p_{kk})\right)^2}.
\end{align*}
The denominator is strictly positive, so the sign of
\(\partial \lambda_k^*/\partial p_{kk}\) is the sign of
\[
(\hatalp_k-\alpha_k)(\alpha_k-A).
\]
Recall that, in this case,
\[
\Sigma_k=p_{kk}\alpha_k+(1-p_{kk})A.
\]

Consider first the case
\[
\hatalp_k<\alpha_k<\Sigma_k .
\]
Then \(\hatalp_k-\alpha_k<0\). Moreover,
\[
\alpha_k< p_{kk}\alpha_k+(1-p_{kk})A
\]
implies \(\alpha_k<A\), since \(p_{kk}\in(0,1)\). Hence
\[
(\hatalp_k-\alpha_k)(\alpha_k-A)>0,
\]
and therefore
\[
\frac{\partial \lambda_k^*}{\partial p_{kk}}>0.
\]

Consider next the case
\[
\Sigma_k<\alpha_k<\hatalp_k .
\]
Then \(\hatalp_k-\alpha_k>0\). Moreover,
\[
p_{kk}\alpha_k+(1-p_{kk})A<\alpha_k
\]
implies \(A<\alpha_k\). Hence
\[
(\hatalp_k-\alpha_k)(\alpha_k-A)>0,
\]
and again
\[
\frac{\partial \lambda_k^*}{\partial p_{kk}}>0.
\]
Thus, whenever \(\lambda_j^*\in(0,1)\) for all groups \(j\), the efficient level of peer
pressure for group \(k\) is increasing in \(p_{kk}\).

We now consider the case in which \(x_j^{sbr}\neq \alpha_j\) for at least one group \(j\).
In this case, \(\Sigma_k\) need not be linear in \(p_{kk}\), because the subjective best
responses of other groups may also vary with \(p_{kk}\) through the equilibrium fixed
point. Define
\[
g(p_{kk})
:=
\sum_{j\neq k}p_{kj}(p_{kk})
\left(q_j\alpha_j+(1-q_j)x_j^{sbr}(p_{kk})\right),
\]
so that
\[
\Sigma_k=p_{kk}\alpha_k+g(p_{kk}).
\]
Then Equation~\eqref{eq_l*_k} becomes
\[
\lambda_k^*
=
\frac{\hatalp_k-\alpha_k}
     {\hatalp_k-p_{kk}\alpha_k-g(p_{kk})}.
\]
Differentiating gives
\[
\frac{\partial \lambda_k^*}{\partial p_{kk}}
=
(\hatalp_k-\alpha_k)
\frac{\alpha_k+g'(p_{kk})}
     {\left(\hatalp_k-p_{kk}\alpha_k-g(p_{kk})\right)^2}.
\]
Thus the sign of the derivative depends on the sign of
\[
(\hatalp_k-\alpha_k)\bigl(\alpha_k+g'(p_{kk})\bigr).
\]
Equivalently, if \(\hatalp_k<\alpha_k\), then
\[
\frac{\partial \lambda_k^*}{\partial p_{kk}}>0
\quad\Longleftrightarrow\quad
\alpha_k+g'(p_{kk})<0,
\]
whereas if \(\hatalp_k>\alpha_k\), then
\[
\frac{\partial \lambda_k^*}{\partial p_{kk}}>0
\quad\Longleftrightarrow\quad
\alpha_k+g'(p_{kk})>0.
\]
Hence, outside the case in which all groups choose the materially optimal effort level,
the effect of \(p_{kk}\) is generally ambiguous.

Finally, consider the effect of the fractions \(q_j\). Since
\[
\lambda_k^*
=
\frac{\hatalp_k-\alpha_k}{\hatalp_k-\Sigma_k},
\]
we have
\[
\frac{\partial \lambda_k^*}{\partial q_j}
=
\frac{\hatalp_k-\alpha_k}{(\hatalp_k-\Sigma_k)^2}
\frac{\partial \Sigma_k}{\partial q_j}.
\]
Therefore the sign of \(\partial \lambda_k^*/\partial q_j\) depends on the sign of
\((\hatalp_k-\alpha_k)\partial\Sigma_k/\partial q_j\).

If for some group $\ell$, \(x_\ell^{sbr}=\alpha_\ell\), then
\[
\bar{x}_\ell=q_\ell\alpha_\ell+(1-q_\ell)\alpha_\ell=\alpha_\ell
\]
and hence \(\Sigma_k=\sum_{\ell}p_{k\ell}\alpha_\ell\) is independent of \(q_{\ell}\). Therefore, for such group $\ell$
\[
\frac{\partial \lambda_k^*}{\partial q_{\ell}}=0
\]

If instead \(x_j^{sbr}\neq \alpha_j\) for some group \(j\), then changes in \(q_j\)
may affect \(\Sigma_k\) both directly, through the average behavior of group \(j\), and
indirectly, through the equilibrium subjective best responses of other groups. In full
generality,
\[
\frac{\partial \Sigma_k}{\partial q_j}
=
p_{kj}
\left[
\alpha_j-x_j^{sbr}
+
(1-q_j)\frac{\partial x_j^{sbr}}{\partial q_j}
\right]
+
\sum_{\ell\neq j}
p_{k\ell}(1-q_\ell)
\frac{\partial x_\ell^{sbr}}{\partial q_j}.
\]
Consequently,
\[
\frac{\partial \lambda_k^*}{\partial q_j}
=
\frac{\hatalp_k-\alpha_k}{(\hatalp_k-\Sigma_k)^2}
\left\{
p_{kj}
\left[
\alpha_j-x_j^{sbr}
+
(1-q_j)\frac{\partial x_j^{sbr}}{\partial q_j}
\right]
+
\sum_{\ell\neq j}
p_{k\ell}(1-q_\ell)
\frac{\partial x_\ell^{sbr}}{\partial q_j}
\right\}.
\]
Since the terms inside braces may be positive or negative, the effect of \(q_j\) on
\(\lambda_k^*\) is generally ambiguous.
\end{proof}

\begin{proof}[Proof of Proposition~\ref{prop:info_sce}]

\text{ }

Suppose that, for a group $k$, $\lambda^*_k\in (0,1)$ as in Equation~\eqref{eq_l*_k} and
$x_k^{sbr}=\alpha_k$. For such group,
\[
E_k[x]=\Sigma_k.
\]
Moreover, by Equation~\eqref{eq_l*_k},
\[
\lambda_k^*
=
\frac{\hatalp_k-\alpha_k}{\hatalp_k-\Sigma_k}.
\]
Substituting this expression into Equation~\eqref{eq_Deltak}, we obtain
\begin{align*}
\Delta_k
&=
(\lambda_k^*)^2(E_k[x]-\hatalp_k)^2 \\
&=
\frac{(\hatalp_k-\alpha_k)^2}{(\hatalp_k-\Sigma_k)^2}
(\Sigma_k-\hatalp_k)^2 \\
&=
(\hatalp_k-\alpha_k)^2.
\end{align*}
\end{proof}

\renewcommand{\thetable}{B\arabic{table}}

\renewcommand{\thefigure}{B\arabic{figure}}

\renewcommand{\thelemma}{B.\arabic{lemma}}

\renewcommand{\thecorollary}{B.\arabic{corollary}}

\renewcommand{\theequation}{B\arabic{equation}}

\section{Explicit results with two groups}\label{secB:res_2G}

This appendix specializes the model to the case of two groups in order to make the main comparative statics fully explicit. The purpose is twofold. First, the results below provide closed-form counterparts to the general characterizations in the main text and clarify the roles of homophily, relative returns, and misspecification. Second, the two-group case allows us to characterize the boundary cases in which the efficient level of peer pressure is corner, i.e.\ $\lambda_k^*\in\{0,1\}$, which are harder to study in the general model. Since there are only two groups, we simplify notation by writing $p_{kk}=p$ and $p_{kj}=1-p$.

We start this subsection by applying the general results to the case $K=2$. Firstly, from Equation~\eqref{eq_xopt_k} and~\eqref{eq_E[x]_k}, we know that
\begin{align}
&x^{sbr}_1 = (1-\lambda_1) \hat \alpha_1 +  \lambda_1 E_1[x]\label{eq_eff_1} \\
&x^{sbr}_2 = (1-\lambda_2) \hat \alpha_2 +  \lambda_2 E_2[x] \label{eq_eff_2} \\
&E_1[x] = p \cdot \left(q_1 \alpha_1 + (1-q_1) x^{sbr}_1 \right) + (1-p)\cdot \left( q_2 \alpha_2 + (1-q_2) x^{sbr}_2\right) \label{eq_Ex1}\\
&E_2[x] = (1-p) \cdot \left(q_1 \alpha_1 + (1-q_1) x^{sbr}_1 \right) + p \cdot \left( q_2 \alpha_2 + (1-q_2) x^{sbr}_2\right)\label{eq_Ex2}
\end{align}

To derive $x^{sbr}_1$ and $x^{sbr}_2$ we need to solve a system of equations with the above equations holding simultaneously. We know from Lemma~\ref{lem_unique_x_k} that the above system admits a unique solution. Specifically, algebraic derivations reveal that
{\footnotesize
$$
x^{sbr}_1 = \frac{\lambda_1(-\hat \alpha_2(\lambda_2-1)(p-1)(q_2-1)+\alpha_1pq_1+\alpha_1\lambda_2(2p-1)q_1(q_2-1)-\alpha_2pq_2+\alpha_2q_2)-\hat \alpha_1(\lambda_1-1)(\lambda_2p(q_2-1)+1)}{\lambda_1p(q_1-1)+\lambda_2(q_2-1)(\lambda_1(2p-1)(q_1-1)+p)+1} 
$$
$$
x^{sbr}_2 = \frac{\lambda_2(-\hat \alpha_1(\lambda_2-1)(p-1)(q_1-1)+\alpha_1(-p)q_1+\alpha_2\lambda_2(2p-1)(q_1-1)q_2+\alpha_2pq_2+\alpha_1q_1)-\hat \alpha_2(\lambda_2-1)(\lambda_2p(q_1-1)+1)}{\lambda_2p(q_1-1)+\lambda_2(q_2-1)(\lambda_2(2p-1)(q_1-1)+p)+1}
$$
}

Solving the maximization problem in Equation~\eqref{eq_pay_gen}, we can derive $\lambda^*_1$ and $\lambda^*_2$. Results are the same as in Equation~\eqref{eq_l*_k}. Specifically,
\begin{equation}\label{eq_lambda*1}
\lambda^*_1 = \frac{\hat \alpha_1-\alpha_1}{\hat \alpha_1-(\alpha_1 p + \alpha_2 (1-p))}    
\end{equation}
\begin{equation}\label{eq_lambda*2}
\lambda^*_2 = \frac{\hat \alpha_2-\alpha_2}{\hat \alpha_2-(\alpha_2 p+\alpha_1 (1-p))}    
\end{equation}

We have previously shown the importance of $\lambda^*_k$ being bounded between $0$ and $1$ and also the conditions for this to happen (see Corollary~\ref{cor_selfconf_gen} and~\ref{cor_l*_k_int}). Trivially, all results in Section~\ref{sec:res} apply for the case with $K=2$ as long as Corollary~\ref{cor_l*_k_int} holds. We proceed the analysis with $K=2$ by studying the conditions under which $\lambda^*_1$ or $\lambda^*_2$ take corner solutions.

\paragraph{Extreme or null peer pressure.} We firstly focus on the conditions for which one single $\lambda^*_j$ is either $0$ or $1$. To find such conditions, we focus on the cases for which $\lambda^*_j=0$ is a lower bound (meaning that the value would have been negative otherwise) or for which $\lambda^*_j=1$ is an upper bound (meaning the value would have been greater than $1$ otherwise).

We begin by deriving the conditions under which $\lambda^*_k \in \{0,1\}$ in the following corollary. 

\begin{corollary}\label{cor_ljextr_K2}

Suppose $K=2$, and consider $j \neq -j \in \{1,2\}$. $\lambda^*_j =1$ if one of the following conditions holds:
\begin{enumerate}
    \item $0 < \alpha_j < \hat{\alpha}_j \leq \alpha_{-j},\ 
    \dfrac{\alpha_{-j} - \hat{\alpha}_j}{\alpha_{-j} - \alpha_j} < p < 1;$
    \item $0 < \alpha_j < \alpha_{-j} < \hat{\alpha}_j;$
    \item $0 < \hat{\alpha}_j \leq \alpha_{-j} < \alpha_j;$
    \item $0<\alpha_{-j} < \hat{\alpha}_j < \alpha_j,\ 
    \dfrac{\alpha_{-j} - \hat{\alpha}_j}{\alpha_{-j} - \alpha_j} < p < 1.$
\end{enumerate}

Similarly, $\lambda^*_j =0$ if one of the following conditions holds:
\begin{enumerate}
    \item $0 < \alpha_j \leq \hat{\alpha}_j < \alpha_{-j},\quad 
    0 < p < \dfrac{\alpha_{-j} - \hat{\alpha}_j}{\alpha_{-j} - \alpha_j};$
    \item $0<\alpha_{-j} < \hat{\alpha}_j \leq \alpha_j,\quad 
    0 < p < \dfrac{\alpha_{-j} - \hat{\alpha}_j}{\alpha_{-j} - \alpha_j}.$
\end{enumerate}

\end{corollary}

Since we obtain $\lambda^*_k$ in Equation~\ref{eq_l*_k} by fixing all $x^{sbr}_k=\alpha_k$, it is clear that it is sufficient that $\lambda^*_k$ admits corner solutions for at least one $k \in K$ for Corollary~\ref{cor_selfconf_gen} to cease to apply.
Since it is computationally too heavy to study the effort levels under corner solutions for $K \geq 3$, the case with $K=2$ becomes particularly useful for conducting this kind of analysis. Interestingly, from Corollary~\ref{cor_ljextr_K2}, we can derive the conditions such that both $\lambda^*_j$ take corner solutions.

\begin{corollary}\label{cor_lextrboth_K2}
Suppose $K=2$, and consider $j \neq -j \in \{1,2\}$. A sufficient condition for $\lambda^*_j \in \{0,1\}$ for all $j \in \{1,2\}$ is 
$$\min\{\alpha_j,\alpha_{-j}\} < \min\{\hat\alpha_j,\hat\alpha_{-j}\} < \max\{\hat\alpha_j,\hat\alpha_{-j}\} < \max\{\alpha_j,\alpha_{-j}\}.$$
\end{corollary}

\paragraph{Effort under extreme or null peer pressure.} Firstly, we consider the case of $\lambda^*_j=0$. Trivially, if $\lambda^*_j=0$, $x^{sbr}_j = \hat \alpha_j$. Thus, if $\lambda^*_1=\lambda^*_2=0$, we obtain a trivial case in which no agent develops peer pressure and all agents play according to their perceived return to effort $\alpha_{\psi,k}$.

We know from Lemma~\ref{lem_unique_laE_k[x]} that, when the efficient level of social
pressure is interior, it satisfies \(x_k^{sbr}=\alpha_k\). By contrast, when the efficient level
is at a corner, \(\lambda_k^*\in\{0,1\}\), the subjective best response need not coincide with
\(\alpha_k\). This is why the corner cases are useful for understanding when the equivalence
between the induced profile and the Nash benchmark breaks down. 

We start from the case when $\lambda^*_j=0$, but $\lambda^*_{-j}=1$. This case can happen under the conditions in Corollary~\ref{cor_lextrboth_K2} and the right assortativity in peer composition.
In this case, \(x^{sbr}_j=\hat\alpha_j\), while \(x^{sbr}_{-j}\) is given by
$$
x^{sbr}_{-j}|_{\lambda^*_j=0,\lambda^*_{-j}=1} = \frac{(1-p)(\alpha_jq_j+\hat \alpha_j(1-q_j))+\alpha_{-j}pq_{-j}}{1-p(1-q_{-j})}
$$

As before, it can be verified that the distance between $x^{sbr}_{-j}$ and $\alpha_{-j}$ decreases in $p$. 

\begin{corollary}\label{cor_|x-a|_lj1l-j0_K2}
Suppose $K=2$, if $\lambda^*_{j}=0$ and $\lambda^*_{-j}=1$, then $|x^{sbr}_{j}|_{\lambda^*_j=0,\lambda^*_{-j}=1}-\alpha_j|$ is constant in $p$, but $|x^{sbr}_{-j}|_{\lambda^*_{j}=0,\lambda^*_{-j}=1}-\alpha_{-j}|$ is decreasing in $p$.
\end{corollary}

\begin{proof}

$\text{ }$

First of all, the conditions such that $\lambda^*_{j}=0$ and $\lambda^*_{-j}=1$ are as follows:
\begin{enumerate}
  \item $0 < \hat{\alpha}_{-j} < \alpha_j < \hat{\alpha}_j < \alpha_{-j}, \quad 0 < p < \frac{\alpha_{-j} - \hat{\alpha}_j}{\alpha_{-j} - \alpha_j}$;
  \item $0 < \alpha_j < \hat{\alpha}_j < \hat{\alpha}_{-j} < \alpha_{-j}, \quad \frac{\alpha_j - \hat{\alpha}_{-j}}{\alpha_j - \alpha_{-j}} < p < \frac{\alpha_{-j} - \hat{\alpha}_j}{\alpha_{-j} - \alpha_j}$;
  \item $0 < \alpha_j < \hat{\alpha}_{-j} < \hat{\alpha}_j < \alpha_{-j}, \quad \frac{\alpha_j - \hat{\alpha}_{-j}}{\alpha_j - \alpha_{-j}} < p < \frac{\alpha_{-j} - \hat{\alpha}_j}{\alpha_{-j} - \alpha_j}$;
  \item $0 < \alpha_{-j} < \hat{\alpha}_j < \hat{\alpha}_{-j} < \alpha_j,
    \quad
    \frac{\alpha_j - \hat{\alpha}_{-j}}{\alpha_j - \alpha_{-j}} < p < \frac{\alpha_{-j} - \hat{\alpha}_j}{\alpha_{-j} - \alpha_j}$;
  \item $0 < \alpha_{-j} < \hat{\alpha}_{-j} < \hat{\alpha}_j < \alpha_j, \quad \frac{\alpha_j - \hat{\alpha}_{-j}}{\alpha_j - \alpha_{-j}}<p<\frac{\alpha_{-j} - \hat{\alpha}_j}{\alpha_{-j} - \alpha_j}$;
  \item $0 < \alpha_{-j} < \hat{\alpha}_j < \alpha_j < \hat{\alpha}_{-j}, \quad 0<p<\frac{\alpha_{-j} - \hat{\alpha}_j}{\alpha_{-j} - \alpha_j}$.
\end{enumerate}
Under \(\lambda_j^*=0\) and \(\lambda_{-j}^*=1\), we have
\[
x^{sbr}_{-j}
=
\frac{(1-p)(\alpha_j q_j+\hat{\alpha}_j(1-q_j))+\alpha_{-j}p q_{-j}}
{1-p(1-q_{-j})}.
\]
Let
\[
B_j:=\alpha_j q_j+\hat{\alpha}_j(1-q_j).
\]
Then
\[
x^{sbr}_{-j}
=
\frac{(1-p)B_j+\alpha_{-j}p q_{-j}}
{1-p(1-q_{-j})}.
\]
Subtracting \(\alpha_{-j}\), we obtain
\[
x^{sbr}_{-j}-\alpha_{-j}
=
\frac{(1-p)(B_j-\alpha_{-j})}
{1-p(1-q_{-j})}.
\]
Since \(q_{-j}\in(0,1)\), the denominator is strictly positive. Moreover,
\[
\frac{\partial}{\partial p}
\left[
\frac{1-p}{1-p(1-q_{-j})}
\right]
=
-\frac{q_{-j}}{\left(1-p(1-q_{-j})\right)^2}<0.
\]
Therefore,
\[
\frac{\partial}{\partial p}
\left(x^{sbr}_{-j}-\alpha_{-j}\right)
=
-\frac{q_{-j}(B_j-\alpha_{-j})}
{\left(1-p(1-q_{-j})\right)^2}.
\]
Hence the derivative of \(x^{sbr}_{-j}-\alpha_{-j}\) has the opposite sign of
\(x^{sbr}_{-j}-\alpha_{-j}\). It follows that
\[
\left|x^{sbr}_{-j}-\alpha_{-j}\right|
\]
is decreasing in \(p\).
\end{proof}

A straightforward observation that one could make is that while both corner solutions are sub-efficient, blind conformism can improve the performance of agents in case they decide to group with their similar (i.e, increase $p$). This result complements Proposition~\ref{prop_cstat_l*_k}. With interior solutions, greater
assortativity raises the efficient intensity of peer pressure. At the corner
\(\lambda_k^*=1\), the intensity of peer pressure cannot increase further; nevertheless,
greater assortativity can still improve behavior by making the peer statistic closer to the
group's true return.

Corollary~\ref{cor_ljextr_K2} and~\ref{cor_lextrboth_K2} highlight the potential cases when $\lambda^*_j = 1$ for all $j \in \{1,2\}$. In such cases, $x^{sbr}_1$ and $x^{sbr}_2$ are symmetric, thus, we can study only one of them, namely
$$
x^{sbr}_j|_{\lambda^*_j=\lambda^*_{-j}=1} = \frac{\alpha_j q_j (p (2 q_{-j}-1)-q_{-j}+1)-\alpha_{-j} (p-1) q_{-j}}{p q_j (2 q_{-j}-1)-p q_{-j}+q_j (-q_{-j})+q_j+q_{-j}}
$$

We can study, again, the distance between $x^{sbr}_j$ and $\alpha_j$ as a proxy for the precision of agents' decisions. First of all, as for previous cases, we can show that such distance is decreasing in $p$ for both $j \in \{1,2\}$. 

\begin{corollary}\label{cor_|x-a|_lj1l-j1_K2}
Suppose $K=2$, if $\lambda^*_j=1$ for all $j \in \{1,2\}$, then, $|x^{sbr}_j|_{\lambda^*_j=\lambda^*_{-j}=1}-\alpha_j|$ is decreasing in $p$.
\end{corollary}

\begin{proof}

$\text{ }$

Similarly to previous proofs, we begin by listing all the conditions such that $\lambda^*_j=1$ for all $j \in \{1,2\}$
\begin{enumerate}
  \item \(\alpha_j < \hat{\alpha}_j \le \alpha_{-j}\).
    \begin{itemize}
      \item \(0 < \hat{\alpha}_{-j} \le \alpha_j\).
      \item \(\alpha_j < \hat{\alpha}_{-j} < \alpha_{-j}, \quad
        \dfrac{\alpha_j - \hat{\alpha}_{-j}}{\alpha_j - \alpha_{-j}} < p < 1.\)
    \end{itemize}
  \item \(\alpha_j < \alpha_{-j} < \hat{\alpha}_j\).
    \begin{itemize}
      \item \(0 < \hat{\alpha}_{-j} \le \alpha_j\).
      \item \(\alpha_j < \hat{\alpha}_{-j} < \alpha_{-j}, \quad
        \dfrac{\alpha_j - \hat{\alpha}_{-j}}{\alpha_j - \alpha_{-j}} < p < 1.\)
    \end{itemize}

  \item \(\hat{\alpha}_j < \alpha_{-j} < \alpha_j\).
    \begin{itemize}
      \item \(0 < \hat{\alpha}_{-j} \le \alpha_j,\quad
        \dfrac{\alpha_j - \hat{\alpha}_{-j}}{\alpha_j - \alpha_{-j}} < p < 1.\)
      \item \(\alpha_j < \hat{\alpha}_{-j}.\)
    \end{itemize}

  \item \(\alpha_{-j} < \hat{\alpha}_j < \alpha_j\).
    \begin{itemize}
      \item \(0 < \hat{\alpha}_{-j} \le \alpha_j,\quad
        \dfrac{\alpha_j - \hat{\alpha}_{-j}}{\alpha_j - \alpha_{-j}} < p < 1.\)
      \item \(\alpha_j < \hat{\alpha}_{-j}.\)
    \end{itemize}
\end{enumerate}
The above conditions can be broken up into 8 conditions, and as for previous statements, it can be verified that under those conditions $|x^{sbr}_j|_{\lambda^*_j=\lambda^*_{-j}=1}-\alpha_j|$ is decreasing in $p$.
\end{proof}

\section{A Generalization to Strictly Concave Payoffs}
\label{secC:general_payoffs}

The analysis in the main text relies on the quadratic specification of material payoffs, under which the subjective best response is linear in the conformity parameter:
\[
x^{sbr}_k=(1-\lambda_k)\hat{\alpha}_k+\lambda_k E_k[x].
\]
This linearity is central in the proof of monotonicity of $x_k^{sbr}(\lambda_k)$ and, in turn, in the characterization of the efficient social pressure level $\lambda_k^*$ in the interior case. Equation \eqref{eq_xopt_k} in the paper delivers the linear best response, while the proof of Lemma \ref{lem_unique_laE_k[x]} exploits the fact that $x_k^{sbr}(\lambda_k)$ is monotone in $\lambda_k$. 

This appendix shows that the linearity in $\lambda_k$ is not essential. It extends the model to a general class of strictly concave single-peaked payoff functions. In that environment, the subjective best response is no longer linear in the primitive parameter $\lambda_k$, but there exists a strictly increasing transformation of $\lambda_k$ into a new effective conformity parameter under which the best response again admits a linear representation. 

Suppose that true material payoffs are given by
\[
\pi_k(x,\alpha_k)=U(\alpha_k-x),
\]
where $U:\mathbb{R}\to\mathbb{R}$ satisfies the following assumptions.

\begin{assumption}\label{ass:U}
The function $U$ is continuously differentiable and strictly concave. Moreover, $U$ has a unique maximizer at $0$.
\end{assumption}

Assumption \ref{ass:U} implies that
\[
U'(z)>0 \quad \text{for } z<0, 
\qquad
U'(0)=0,
\qquad
U'(z)<0 \quad \text{for } z>0.
\]
Thus, $U$ is strictly increasing on $(-\infty,0)$ and strictly decreasing on $(0,\infty)$.

For a misspecified agent in group $k$ with perceived return $\hat{\alpha}_k$, define subjective utility as
\[
u_k(x;\lambda)
=
(1-\lambda)U(\hat{\alpha}_k-x)+\lambda U(E_k[x]-x),
\qquad \lambda\in[0,1].
\]
This is the natural analogue of the subjective utility used in the paper, replacing the quadratic form with a general strictly concave single-peaked function. In this general case, misspecified agents know the functional form of $U$ up to a constant, in the same way as in the quadratic case they do not know $C$.

Define the subjective best response by
\[
x_k^{sbr}(\lambda)\in \arg\max_{x\in \mathbb{R}_+} u_k(x;\lambda).
\]

We now proceed in two steps. We first study the individual optimization problem conditional on a fixed reference term \(E_k[x]\). This yields a representation result for the subjective best response. We then return to equilibrium and show that, at an interior efficient allocation, the implied value of \(E_k[x]\) is consistent with the equilibrium profile.

\begin{proposition}
\label{prop:concave_representation}
Suppose Assumption 1 holds and suppose \(E_k[x]\neq \hat{\alpha}_k\). Fix \(E_k[x]\) in the individual optimization problem. Then, for every \(\lambda\in[0,1]\), there exists a unique
\[
\tau_k(\lambda)\in[0,1]
\]
such that
\[
x_k^{sbr}(\lambda)
=
(1-\tau_k(\lambda))\hat{\alpha}_k+\tau_k(\lambda)E_k[x].
\]
Moreover, \(\tau_k:[0,1]\to[0,1]\) is continuous and strictly increasing, with
\[
\tau_k(0)=0,
\qquad
\tau_k(1)=1.
\]
\end{proposition}

\begin{proof}

$\text{ }$

Fix \(E_k[x]\neq \hat{\alpha}_k\). Since \(U\) is strictly concave, \(u_k(\cdot;\lambda)\) is strictly concave for every \(\lambda\in(0,1)\), so the subjective best response \(x_k^{sbr}(\lambda)\) is unique. Moreover,
\[
x_k^{sbr}(0)=\hat{\alpha}_k,
\qquad
x_k^{sbr}(1)=E_k[x].
\]

For \(\lambda\in(0,1)\), the first-order condition is
\[
(1-\lambda)U'(\hat{\alpha}_k-x)+\lambda U'(E_k[x]-x)=0.
\]
If \(E_k[x]>\hat{\alpha}_k\), the left-hand side is strictly positive at \(x=\hat{\alpha}_k\) and strictly negative at \(x=E_k[x]\). Hence the unique maximizer lies in \((\hat{\alpha}_k,E_k[x])\). If \(E_k[x]<\hat{\alpha}_k\), the same argument implies
\[
x_k^{sbr}(\lambda)\in(E_k[x],\hat{\alpha}_k).
\]
Therefore, for every \(\lambda\in[0,1]\), there exists a unique \(\tau_k(\lambda)\in[0,1]\) such that
\[
x_k^{sbr}(\lambda)
=
(1-\tau_k(\lambda))\hat{\alpha}_k+\tau_k(\lambda)E_k[x].
\]


To show that $\tau_k$ is strictly increasing, it suffices to show that
$x_k^{sbr}(\lambda)$ moves strictly from $\hat{\alpha}_k$ toward $E_k[x]$ as $\lambda$ increases.

Fix $\lambda''>\lambda'$. Suppose first that $E_k[x]>\hat{\alpha}_k$. Then
$x_k^{sbr}(\lambda')\in(\hat{\alpha}_k,E_k[x])$, so
\[
U'(\hat{\alpha}_k-x_k^{sbr}(\lambda'))>0
\quad\text{and}\quad
U'(E_k[x]-x_k^{sbr}(\lambda'))<0.
\]
Now recall that
\[
\partial_x u_k(x;\lambda)
=
-(1-\lambda)U'(\hat{\alpha}_k-x)-\lambda U'(E_k[x]-x).
\]
Evaluating $\partial_x u_k(\cdot;\lambda)$ at $x_k^{sbr}(\lambda')$, the first term is
negative and the second term is positive. Hence, increasing $\lambda$ reduces the weight
on the negative term and increases the weight on the positive term. Therefore,
\(
\partial_x u_k(x_k^{sbr}(\lambda');\lambda'')
>
\partial_x u_k(x_k^{sbr}(\lambda');\lambda')
\).

Since $x_k^{sbr}(\lambda')$ is the maximizer at $\lambda'$, the first-order condition gives
\(
\partial_x u_k(x_k^{sbr}(\lambda');\lambda')=0
\).
Therefore,
\(
\partial_x u_k(x_k^{sbr}(\lambda');\lambda'')>0
\).

Since $u_k(\cdot;\lambda'')$ is strictly concave, its derivative is positive to the left of the
unique maximizer and negative to the right. It follows that
\(
x_k^{sbr}(\lambda'')>x_k^{sbr}(\lambda')
\).

If instead $E_k[x]<\hat{\alpha}_k$, the same argument gives
\(
x_k^{sbr}(\lambda'')<x_k^{sbr}(\lambda')
\).
Thus $x_k^{sbr}(\lambda)$ moves strictly from $\hat{\alpha}_k$ toward $E_k[x]$ as $\lambda$
increases, and therefore $\tau_k$ is strictly increasing.

Finally, continuity of $\tau_k$ follows from continuity of the unique maximizer
$x_k^{sbr}(\lambda)$ in $\lambda$, and the endpoint conditions $\tau_k(0)=0$ and
$\tau_k(1)=1$ follow from the identities above.
\end{proof}

Proposition \ref{prop:concave_representation} shows that, although the subjective best response is not generally linear in the primitive parameter \(\lambda\), it is linear in the transformed parameter \(\tau_k(\lambda)\). In the quadratic benchmark, \(\tau_k(\lambda)=\lambda\), so the analysis in the main text is recovered as a special case.

This observation is sufficient to recover the analogue of Corollary \ref{cor_selfconf_gen} in the general strictly concave case. In particular, when the efficient social pressure level is interior, the relevant object is the transformed parameter
\[
\tilde{\lambda}_k^*:=\tau_k(\lambda_k^*).
\]

The previous proposition is an individual representation result, obtained by holding the
reference term \(E_k[x]\) fixed. It is nevertheless sufficient to recover the interior benchmark
of the quadratic model. Indeed, at an interior efficient profile, material optimality requires
\(x_k^{sbr}=\alpha_k\) for every group \(k\). Hence the equilibrium peer statistic is pinned down
by true returns:
\[
E_k[x]=\sum_{j=1}^K p_{kj}\alpha_j.
\]
Thus, once the model is written in terms of the transformed conformity parameter
\(\tilde\lambda_k^*:=\tau_k(\lambda_k^*)\), the same interior characterization as in the quadratic
case is recovered.

\begin{corollary}
\label{cor:concave_interior}
Suppose the conditions of Proposition \ref{prop:concave_representation} hold, and suppose the efficient social pressure level is interior for all misspecified groups. 
Then
\[
\tilde{\lambda}_k^*
=
\frac{\hat{\alpha}_k-\alpha_k}
{\hat{\alpha}_k-\sum_{j=1}^K p_{kj}\alpha_j}.
\]
Consequently:

\begin{enumerate}
    \item \(\lambda_k^*\in(0,1)\) if and only if \(\tilde{\lambda}_k^*\in(0,1)\), which holds if and only if one of the following two conditions is satisfied:
    \[
    \hat{\alpha}_k<\alpha_k<\sum_{j=1}^K p_{kj}\alpha_j,
    \qquad\text{or}\qquad
    \sum_{j=1}^K p_{kj}\alpha_j<\alpha_k<\hat{\alpha}_k.
    \]

    \item The transformed efficient social pressure \(\tilde{\lambda}_k^*\) does not depend on the frequencies \(q_j\).

    \item If \(\tilde{\lambda}_k^*\in(0,1)\), then \(\tilde{\lambda}_k^*\) is increasing in \(\hat{\alpha}_k\) if and only if
    \[
    \alpha_k>\sum_{j=1}^K p_{kj}\alpha_j.
    \]
\item Under the proportional variation of out-group weights used in
Proposition~\ref{prop_cstat_l*_k}, if \(\tilde{\lambda}_j^*\in(0,1)\) for all
misspecified groups \(j\), then \(\tilde{\lambda}_k^*\) is increasing in \(p_{kk}\).
\end{enumerate}
\end{corollary}

\begin{proof}

$\text{ }$

Since \(U(\alpha_k-x)\) is uniquely maximized at \(x=\alpha_k\), any interior efficient social pressure level must satisfy
\[
x_k^{sbr}(\lambda_k^*)=\alpha_k.
\]
Hence, exactly as in Corollary~\ref{cor_selfconf_gen} for the quadratic case, the interior efficiency condition implies that the agent's subjective best response must equal the true benchmark. Combining this with Proposition~\ref{prop:concave_representation} yields
\[
\alpha_k=(1-\tilde{\lambda}_k^*)\hat{\alpha}_k+\tilde{\lambda}_k^* E_k[x].
\]
At the interior efficient profile, behavior coincides with the Nash profile, so \(x_j^{sbr}=\alpha_j\) for every \(j\), and therefore
\[
E_k[x]=\sum_{j=1}^K p_{kj}\alpha_j.
\]
Substituting yields
\[
\tilde{\lambda}_k^*
=
\frac{\hat{\alpha}_k-\alpha_k}
{\hat{\alpha}_k-\sum_{j=1}^K p_{kj}\alpha_j}.
\]
The characterization of interiority follows immediately from the condition
\(\tilde{\lambda}_k^*\in(0,1)\), and since \(\tau_k\) is strictly increasing with
\(\tau_k(0)=0\) and \(\tau_k(1)=1\), this is equivalent to \(\lambda_k^*\in(0,1)\).

The fact that \(\tilde{\lambda}_k^*\) does not depend on the \(q_j\)'s is immediate from the formula above. 
The comparative statics with respect to \(\hat{\alpha}_k\) are exactly the same as in the
quadratic benchmark. The comparative static with respect to \(p_{kk}\) is also the same
under the proportional variation of out-group weights used in
Proposition~\ref{prop_cstat_l*_k}.
\end{proof}

Corollary~\ref{cor:concave_interior} shows that the interiority condition and the main
comparative-statics insights of the quadratic model carry over to the transformed parameter
\(\tilde{\lambda}_k^*\). In particular, the analogue of Proposition~\ref{prop_cstat_l*_k}
follows for \(\tilde{\lambda}_k^*\), since its expression coincides with Equation~\eqref{eq_l*_k},
subject to the same qualifications on interiority and peer-composition changes. Thus, beyond
the quadratic benchmark, what is lost is the explicit representation of the efficient social
pressure level in terms of the primitive parameter \(\lambda_k^*\), nor the structure of the
interiority condition nor the associated comparative statics expressed in terms of
\(\tilde{\lambda}_k^*\). By contrast, the closed-form expression for the perceived value of
monitoring in Proposition~\ref{prop:info_sce} is specific to the quadratic specification.

\end{appendices}
\end{document}